\documentclass[letterpaper, 10pt, conference]{ieeeconf}
\usepackage[autostyle]{csquotes}
\usepackage{bm}
\usepackage{amssymb}
\usepackage{epsfig}
\usepackage{subfigure}
\usepackage{multicol, blindtext}
\usepackage{amsmath}
\usepackage{graphicx}
\usepackage{epstopdf}
\usepackage{float}
\usepackage{amsfonts}
\usepackage{color}
\usepackage{dblfloatfix}
\usepackage{multirow}%
\usepackage{fourier} 
\usepackage{array}
\usepackage{makecell}

\providecommand{\U}[1]{\protect\rule{.1in}{.1in}}

\newtheorem{theorem}{\rm\textbf{Theorem}}

\IEEEoverridecommandlockouts
\begin{document}

\title{{\LARGE \textbf{Optimal Control of Connected Automated Vehicles with Event/Self-Triggered Control Barrier Functions }}
\thanks{This work was supported in part by NSF under grants ECCS-1931600,
DMS-1664644, CNS-1645681, CNS-2149511, by AFOSR under grant FA9550-19-1-0158,
by ARPA-E under grant DE-AR0001282, by the MathWorks, and by NPRP grant
(12S-0228-190177) from the Qatar National Research Fund, a member of
the Qatar Foundation (the statements made herein are solely the responsibility
of the authors).
The authors are with the Division of Systems Engineering and Center for
Information and Systems Engineering, Boston University, Brookline, MA,
02446, USA, CSAIL, MIT, USA and Electrical Engineering Department,Qatar University, Doha, Qatar \texttt{{\small \{esabouni, cgc\}@bu.edu, \{weixy@mit.edu\}, \{nader.meskin@qu.edu.qa\}}}}}
\author{Ehsan Sabouni, Christos G. Cassandras, Wei Xiao and Nader Meskin}
\maketitle

\begin{abstract}
We address the problem of controlling Connected and Automated Vehicles (CAVs) in conflict areas of a traffic network subject to hard safety constraints. It has been shown that such problems can be solved through a combination of tractable optimal control problem formulations and the use of Control Barrier Functions (CBFs) that guarantee the satisfaction of all constraints. These solutions can be reduced to a sequence of Quadratic Programs (QPs) which are efficiently solved on-line over discrete time steps. However, the feasibility of each such QP cannot be guaranteed over every time step. To overcome this limitation, we develop both an event-triggered approach and a self-triggered approach such that the next QP is triggered by properly defined events. We show that both approaches, each in a different way, eliminate infeasible cases due to time-driven inter-sampling effects, thus also eliminating the need for selecting the size of time steps.
Simulation examples are included to compare the two new schemes and to illustrate how overall infeasibilities can be significantly reduced while at the same time reducing the need for communication among CAVs without compromising performance.
\end{abstract}

\thispagestyle{empty} \pagestyle{empty}


\section{INTRODUCTION}

The emergence of Connected and Automated Vehicles (CAVs) along with new traffic infrastructure technologies \cite{li2013survey},\cite{9625017} over the past decade have brought the promise of resolving long-lasting problems in transportation networks such as accidents, congestion, and unsustainable energy consumption along with environmental pollution \cite{deWaard09},\cite{Schrank20152015UM},\cite{kavalchuk2020performance}. Meeting this goal heavily depends on effective traffic management, specifically at the bottleneck points of a transportation network such as intersections, roundabouts, and merging roadways \cite{VANDENBERG201643}. 

To date, both centralized and decentralized methods have been proposed to tackle the control and coordination problem of CAVs in conflict areas; an overview of such methods may be found in \cite{7562449}. Platoon formation \cite{xu2019grouping}, \cite{wu2013mathematical}, \cite{rajamani2000demonstration} and reservation-based methods \cite{1373519},\cite{au2010motion},\cite{zhang2013analysis} are among the centralized approaches, which are limited by the need for powerful central computation resources and are typically prone to disturbances and security threats.
In contrast, in decentralized methods each CAV is responsible for its own on-board computation with information from other vehicles limited to a set of neighbors \cite{7313484}. Constrained optimal control problems can then be formulated with objectives usually involving minimizing acceleration or maximizing passenger comfort (measured as the acceleration derivative or jerk), or jointly minimizing travel time through conflict areas and energy consumption. These problems can be analytically solved in some cases, e.g., for optimal merging \cite{XIAO2021109333} or crossing
a signal-free intersection \cite{Zhang2018}.
However, obtaining such solutions becomes computationally prohibitive for real-time applications when an optimal trajectory involves multiple constraints becoming active. Thus, on-line control methods such as Model Predictive Control (MPC) techniques or Control Barrier function (CBFs) are often adopted for the handling of additional constraints.

In the MPC approach proposed in  \cite{garcia1989model}, the time is normally discretized and an optimization problem is solved at each time instant with the addition of appropriate inequality constraints; then, the system dynamics are updated. Since both control and state are considered as the decision variables in the optimization problem, MPC is very effective for problems with simple (usually linear or linearized) dynamics, objectives, and constraints \cite{cao2015cooperative}. Alternatively, CBFs \cite{Xiao2019} \cite{CBF_QP(2017)} can overcome some shortcomings of the MPC method \cite{mukai2017model} as they do not need states as decision variables, instead mapping state constraints onto new ones that only involve the decision variables in a linear fashion. Moreover, CBFs can be used with nonlinear (affine in control) system dynamics and they
have a crucial forward invariance property which guarantees the satisfaction of safety constraints over all time as long as these constraints are initially satisfied.



An approach combining optimal control solutions with CBFs was recently presented in \cite{XIAO2021109592}. In this combined approach (termed OCBF), the solution of an \emph{unconstrained} optimal control problem is first derived and used as a reference control. Then, the resulting control reference trajectory is optimally tracked subject to
a set of CBF constraints which ensure the satisfaction of all constraints of the original optimal control problem. Finally, this optimal tracking problem is efficiently solved by discretizing time and solving a simple Quadratic Problem (QP) at each discrete time step over which the control input is held constant \cite{CBF_QP(2017)}. The use of CBFs in this approach exploits their forward invariance property to guarantee that all constraints they enforce are satisfied at all times if they are initially satisfied. In addition, CBFs are designed to impose \emph{linear} constraints on the control which is what enables the efficient solution of the tracking problem through a sequence of QPs. This approach can also be shown to provide additional flexibility in terms of using nonlinear vehicle dynamics (as long as they are affine in the control), complex objective functions, and tolerate process and measurement noise \cite{XIAO2021109592}. 

However, in solving a sequence of QPs the control update interval in the time discretization process must be sufficiently small in order to always guarantee that every QP is feasible. In practice, such feasibility can be often seen to be violated due to the fact that it is extremely difficult to pick a proper discretization time which can be guaranteed to always work. In this paper, an \emph{event-triggered}  and a \emph{self-triggered} approach are considered as two solutions to remedy this issue. We note that the idea of synthesizing event-triggered controllers and BFs or Lyapunov functions has been used in \cite{ong2018event} with the goal of improving stability, while
a unified event-driven scheme is proposed in \cite{taylor2020safety} with an Input-to-State barrier function to impose safety under an input disturbance. 

The contribution of this paper is to replace the \emph{time-driven} nature of the discretization process used in the OCBF approach which involves a sequence of QPs by an \emph{event-driven} mechanism, hence achieving QP feasibility independent of a time step choice. In the event-triggering scheme, given the system state at the start of a given QP instance, we extend the approach introduced in \cite{Xiao2021EventTriggeredSC} for a multi-agent system to define events associated with the states of CAVs reaching a certain bound, at which point the next QP instance is triggered. 
On the other hand, in the self-triggering scheme, we provide a minimum inter-event time guarantee by predicting the first time instant that any of the CBF constraints in the QP problem is violated, hence we can determine the triggering time for the next QP instance. 
Both methods provide a guarantee for the forward invariance property of CBFs and eliminate infeasible cases due to time-driven inter-sampling effects (additional infeasibilities are still possible due to potentially conflicting constraints within a QP; this separate issue has been addressed in \cite{XIAO2022inf}). 

The advantages of these event-driven schemes can be summarized as follows:
$(i)$ Infeasible QP instances due to inter-sampling effects are eliminated,
$(ii)$ There is no longer a need to determine a proper time step size required in the time-driven methods,
$(iii)$ The number of control updates under event-driven schemes is generally reduced, thereby reducing the overall computational cost, and
$(iv)$ Since the number of QPs that need to be solved is reduced, this also reduces 
the need for unnecessary communication among CAVs. This reduced need for communication, combined with the unpredictability of event-triggering relative to a fixed time discretization approach, results in the system being less susceptible to malicious attacks.

The paper is organized as follows. In Section II, we provide an overview of the decentralized constrained optimal control for CAVs in any conflict area setting, along with a brief review of CBFs to set the stage for the OCBF approach. We also review the time-driven approach for solving such optimal control problems, motivating the proposed solutions to the problem. In Section III, both methods are separately presented, including the formulation and solution of QPs in both frameworks. In Section V, simulation results compare time-driven, event-triggered, and self-triggered schemes, in terms of their performance metrics, computational load, and infeasible cases to show how constraint violations can be reduced through the proposed approaches.

\section{Problem Formulation and Time-Driven Control Solutions}

\label{sec:problem}
In this section, we review the setting for CAVs whose motion is cooperatively controlled at conflict areas of a traffic network. This includes merging roads, signal-free intersections, roundabouts, and highway segments where lane change maneuvers take place. We define a Control Zone (CZ) to be an area within which CAVs can communicate with each other or with a coordinator (e.g., a Road-Side Unit (RSU)) which is responsible for facilitating the exchange of information (but not control individual vehicles) within this CZ. As an example, Fig. \ref{fig:merging} shows a conflict area due to vehicles merging from two single-lane roads and there is a single Merging Point (MP) which vehicles must cross from either road \cite{XIAO2021109333}. 

In such a setting, assuming all traffic consists of CAVs, a finite horizon constrained optimal control problem can be formulated aiming to determine trajectories that jointly minimize travel time and energy consumption through the CZ while also ensuring passenger comfort (by minimizing jerk or centrifugal forces) and guaranteeing safety constraints are always satisfied. 
Let $F(t)$ be the set of indices of all CAVs located in the
CZ at time $t$. A CAV enters the CZ at one of several origins (e.g., $O$ and $O'$ in Fig. \ref{fig:merging}) and leaves at one of possibly several exit points (e.g., $M$ in Fig. \ref{fig:merging}). 
The index $0$ is used to denote a CAV that has just left the CZ. Let $N(t)$ be the
cardinality of $F(t)$. Thus, if a CAV arrives at time $t$, it is assigned the
index $N(t)+1$. All CAV indices in $F(t)$ decrease by one when a CAV passes over
the MP and the vehicle whose index is $-1$ is dropped.

The vehicle dynamics for each CAV $i\in F(t)$ along the lane to which it
belongs in a given CZ are assumed to be of the form
\begin{equation} \label{VehicleDynamics}
\left[
\begin{array}
[c]{c}%
\dot{x}_{i}(t)\\
\dot{v}_{i}(t)
\end{array}
\right]  =\left[
\begin{array}
[c]{c}%
v_{i}(t)\\
u_{i}(t)
\end{array}
\right],  
\end{equation}
where $x_{i}(t)$ denotes the distance from the origin at which CAV $i$ arrives, $v_{i}(t)$ denotes the velocity, and $u_{i}(t)$ denotes the control input (acceleration).
There are two objectives for each CAV, as detailed next.\\
\begin{figure}[H]
\vspace{-5mm}
\centering
$\hspace{-4mm}$\includegraphics[scale=0.8]{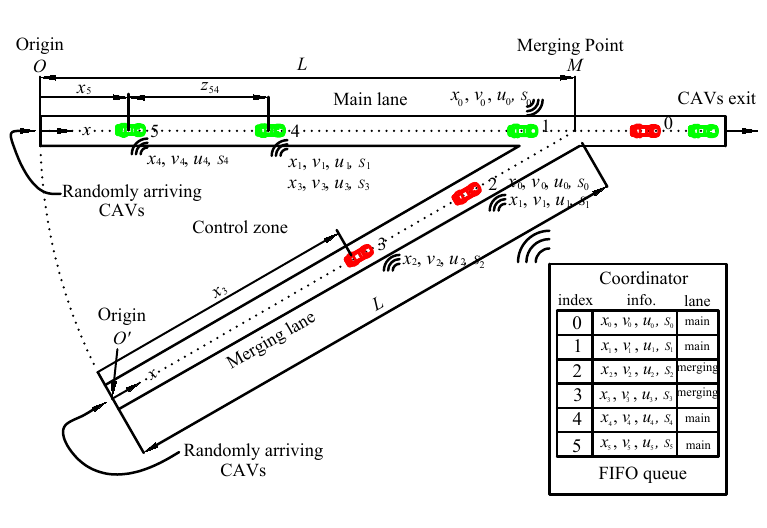} \caption{The merging problem}%
\label{fig:merging}%
\end{figure}
\noindent {\bf Objective 1} (Minimize travel time): Let $t_{i}^{0}$ and $t_{i}^{f}$
denote the time that CAV $i\in F(t)$ arrives at its origin
and leaves the CZ at its exit point, respectively. We wish to minimize the travel time
$t_{i}^{f}-t_{i}^{0}$ for CAV $i$.\\
{\bf Objective 2} (Minimize energy consumption): We also wish to minimize
the energy consumption for each CAV $i$:
\begin{equation}
J_{i}(u_{i}(t),t_{i}^{f})=\int_{t_{i}^{0}}^{t_{i}^{f}}\mathcal{L}_i(|u_{i}(t)|)dt,
\end{equation}
where $\mathcal{L}_i(\cdot)$ is a strictly increasing function of its argument.\\
{\bf Constraint 1} (Safety constraints): Let $i_{p}$ denote the index of
the CAV which physically immediately precedes $i$ in the CZ (if one is
present). We require that the distance $z_{i,i_{p}}(t):=x_{i_{p}}(t)-x_{i}(t)$
be constrained by:
\begin{equation}
z_{i,i_{p}}(t)\geq\varphi v_{i}(t)+\delta,\text{ \ }\forall t\in\lbrack
t_{i}^{0},t_{i}^{f}], \label{Safety}%
\end{equation}
where $\varphi$ denotes the reaction time (as a rule, $\varphi=1.8s$ is used,
e.g., \cite{Vogel2003}) and $\delta$ is a given minimum safe distance.
 If we define $z_{i,i_{p}}$ to be the distance from
the center of CAV $i$ to the center of CAV $i_{p}$, then $\delta$ depends on the length of these two CAVs (generally dependent on
$i$ and $i_{p}$ but taken to be a constant over all CAVs for simplicity).\\
{\bf Constraint 2} (Safe merging): Whenever a CAV crosses a MP, a lateral collision is possible and there must be adequate safe space for the CAV at this MP to avoid such collision, i.e.,
\begin{equation}
\label{SafeMerging}
z_{i,i_c}(t_{i}^{m})\geq\varphi v_{i}(t_{i}^{m})+\delta,
\end{equation}
where $i_c$ is the index of the CAV that may collide with CAV $i$ at merging point $m=\lbrace 1,...,n_i \rbrace$ where $n_i$ is the total number of MPs that CAV $i$ passes in the CZ. The determination of CAV $i_c$ depends on the policy adopted for sequencing CAVs through the CZ, such as First-In-First-Out (FIFO) based on the arrival times of CAVs, or any other desired policy. It is worth noting that this constraint only applies at a certain time $t_{i}^{m}$ which obviously depends on how the CAVs are controlled. As an example, in Fig. \ref{fig:merging} under FIFO, we have $i_c=i-1$ and $t_i^m=t_i^f$ since the MP defines the exit from the CZ.\\
{\bf Constraint 3} (Vehicle limitations): Finally, there are constraints
on the speed and acceleration for each $i\in F(t)$:
\begin{equation}
\begin{aligned} v_{\min} \leq v_i(t)\leq v_{\max}, \forall t\in[t_i^0,t_i^f]\end{aligned} \label{VehicleConstraints1}%
\end{equation}
\begin{equation}
\begin{aligned} u_{\min}\leq u_i(t)\leq u_{\max}, \forall t\in[t_i^0,t_i^f],\end{aligned} \label{VehicleConstraints2}%
\end{equation}
where $v_{\max}> 0$ and $v_{\min} \geq 0$ denote the maximum and minimum speed allowed
in the CZ for CAV $i$, $u_{{\min}}<0$ and $u_{\max}>0$ denote the minimum and maximum
control for CAV $i$, respectively.\\
\textbf{Optimal Control Problem formulation.} Our goal is to determine a control law achieving objectives 1-2 subject to constraints 1-3 for each $i \in F(t)$ governed by the dynamics (\ref{VehicleDynamics}). Choosing $\mathcal{L}_i(u_i(t))=\frac{1}{2}u_i^2(t)$ and normalizing travel time and $\frac{1}{2}u_{i}^{2}(t)$, we use the weight $\alpha\in[0,1]$ to construct
a convex combination as follows:
\begin{equation}\label{eqn:energyobja_m}
\begin{aligned}\min_{u_{i}(t),t_i^f} J_i(u_i(t),t_i^f)= \int_{t_i^0}^{t_i^f}\left(\alpha + \frac{(1-\alpha)\frac{1}{2}u_i^2(t)}{\frac{1}{2}\max \{u_{\max}^2, u_{\min}^2\}}\right)dt \end{aligned}.
\end{equation}
 Letting $\beta:=\frac{\alpha\max\{u_{\max}^{2},u_{\min}^{2}\}}{2(1-\alpha)}$, we obtain a simplified form: 
\begin{equation}\label{eqn:energyobja}
\min_{u_{i}(t),t_i^f}J_{i}(u_{i}(t),t_i^f):=\beta(t_{i}^{f}-t_{i}^{0})+\int_{t_{i}^{0}%
}^{t_{i}^{f}}\frac{1}{2}u_{i}^{2}(t)dt,
\end{equation}
where $\beta\geq0$ is an adjustable weight to penalize
travel time relative to the energy cost. Note that the solution is \emph{decentralized} in the sense that CAV $i$ requires information only from CAVs $i_p$ and $i_c$ required in (\ref{Safety}) and (\ref{SafeMerging}).

Problem (\ref{eqn:energyobja}) subject to (\ref{VehicleDynamics}), (\ref{Safety}), (\ref{SafeMerging}), (\ref{VehicleConstraints1})
and (\ref{VehicleConstraints2}) can be analytically solved in some cases, e.g., the merging problem in Fig. \ref{fig:merging}
\cite{XIAO2021109333}
and a signal-free intersection 
\cite{Zhang2018}.
However, obtaining solutions for real-time applications becomes prohibitive when an optimal trajectory involves multiple constraints becoming active. This has motivated an approach which combines a solution of the unconstrained problem (\ref{eqn:energyobja}), which can be obtained very fast, with the use of Control Barrier Functions (CBFs) which provide guarantees that (\ref{Safety}), (\ref{SafeMerging}), (\ref{VehicleConstraints1}) and (\ref{VehicleConstraints2}) are always satisfied through constraints that are linear in the control, thus rendering solutions to this alternative problem obtainable by solving a sequence of computationally efficient QPs. This approach is termed Optimal Control with Control Barrier Functions (OCBF) \cite{XIAO2021109592}.

\textbf{The OCBF approach.} The OCBF approach consists of three steps: $(i)$ the solution of the \emph{unconstrained} optimal control problem (\ref{eqn:energyobja}) is used as a reference control, $(ii)$ the resulting control reference trajectory is optimally tracked subject to the constraint \eqref{VehicleConstraints2}, as well as a set of CBF constraints enforcing (\ref{Safety}), (\ref{SafeMerging}) and \eqref{VehicleConstraints1}. $(iii)$ This optimal tracking problem is efficiently solved by discretizing time and solving a simple QP at each discrete time step. The significance of CBFs in this approach is twofold: first, their forward invariance property \cite{XIAO2021109592} guarantees that all constraints they enforce are satisfied at all times if they are initially satisfied; second, CBFs impose \emph{linear} constraints on the control which is what enables the efficient solution of the tracking problem through the sequence of QPs in $(iii)$ above.

The reference control in step $(i)$ above is denoted by $u_{i}^{\textrm{ref}}(t)$. The unconstrained solution to (\ref{eqn:energyobja}) is denoted by $u_i^*(t)$, thus we usually set $u_{i}^{\textrm{ref}}(t)=u_i^*(t)$. However, $u_{i}^{\textrm{ref}}(t)$ may be chosen to be any desired control trajectory and, in general, we use $u_{i}^{\textrm{ref}}(t)=h(u_i^*(t),x_i^*(t), \textbf{x}_i(t))$ where $\textbf{x}_i(t)\equiv(x_i(t),v_i(t)), ~\textbf{x}_i \in \textbf{X}$  ($ \mathbf{X} \subset \mathbb{R}^2$ is the state space). Thus, in addition to the unconstrained optimal control and position $u_i^*(t),x_i^*(t)$, observations of the actual CAV state $\textbf{x}_i(t)$ provide direct feedback as well. 

To derive the CBFs that ensure the constraints (\ref{Safety}), (\ref{SafeMerging}), and (\ref{VehicleConstraints1}) are always satisfied, we use the vehicle dynamics (\ref{VehicleDynamics}) to define $f(\textbf{x}_i(t))=[v_i(t),0]^T$ and $g(\textbf{x}_i(t))=[0,1]^T$. Each of these constraints can be easily written in the form of $b_q(\textbf{x}(t)) \geq 0$, $q \in \lbrace  1,...,n \rbrace$ where $n$ stands for the number of constraints and $\mathbf{x}(t)=[\mathbf{x}_1(t),\mathbf{x}_2(t),...,\mathbf{x}_{N(t)}(t)]$. The CBF method (details provided in \cite{XIAO2021109592}) maps a constraint $b_q(\textbf{x}(t)) \geq 0$ onto a new constraint which is linear in the control input $u_i(t)$ and takes the general form 
\begin{equation} \label{CBF general constraint}
L_fb_q(\textbf{x}(t))+L_gb_q(\textbf{x}(t))u_i(t)+\gamma( b_q(\textbf{x}(t))) \geq 0,
\end{equation}
where $L_f,L_g$ denote the Lie derivatives of $b_q(\textbf{x}(t))$ along $f$ and $g$, respectively and $\gamma(\cdot)$ stands for any class-$\mathcal{K}$ function \cite{XIAO2021109592}. It has been established 
\cite{XIAO2021109592}
that satisfaction of (\ref{CBF general constraint}) implies the satisfaction of the original problem constraint $b_q(\textbf{x}(t)) \geq 0$ because of the forward invariance property. It is worth observing that the newly obtained constraints are \emph{sufficient} conditions for the original problem constraints, therefore, potentially conservative. 

We now apply (\ref{CBF general constraint}) to obtain the CBF constraint associated with the safety constraint (\ref{Safety}). By setting
\begin{align} \label{b1}
    b_1(\textbf{x}_i(t),\textbf{x}_{i_p}(t))&=z_{i,i_{p}}(t)-\varphi v_{i}(t)-\delta \nonumber\\
    &=x_{i_p}(t)-x_i(t)-\varphi v_i(t)-\delta,
\end{align} and since $b_1(\textbf{x}_i(t),\textbf{x}_{i_p}(t))$ is differentiable,
the CBF constraint for (\ref{Safety}) is
\begin{equation}\label{CBF1}\small
\underbrace{v_{i_p}(t)-v_i(t)}_{L_fb_1(\textbf{x}_i(t),\textbf{x}_{i_p}(t))}+\underbrace{-\varphi}_{L_gb_1(\textbf{x}_i(t))} u_i(t)+\underbrace{k_1(z_{i,i_p}(t)-\varphi v_i(t)-\delta)}_{\gamma_1(b_1(\textbf{x}_i(t),\textbf{x}_{i_p}(t)))} \geq 0,
\end{equation}
where the class-$\mathcal{K}$ function $\gamma(x)=k_1x$ is chosen here to be linear.


Deriving the CBF constraint for the safe merging constraint (\ref{SafeMerging}) poses a technical challenge due to the fact that it only applies at a certain time $t_i^{m}$, whereas a CBF is required to be in a continuously differentiable form. To tackle this problem. we apply a technique used in \cite{XIAO2021109592} to convert (\ref{SafeMerging}) to a continuous differentiable form as follows:
\begin{equation}
z_{i,i_c}(t)-\Phi(x_i(t)) v_{i}(t)-\delta \geq 0, \   {\ }\forall t\in[t_i^0,t_i^{m}],
\end{equation}
where $\Phi : \mathbb{R} \rightarrow \mathbb{R}$ may be any continuously differentiable function as long as it is strictly increasing and satisfies the boundary conditions $\Phi(x_i(t_i^0))=0$ and $\Phi(x_i(t_i^{m}))=\varphi$. In this case, a linear function can satisfy both conditions:
\begin{equation}
    \Phi(x_i(t))=\varphi  \frac{x_i(t)}{L},
\end{equation}
where $L$ is the length of road traveled by the CAV from its entry to the CZ to the MP of interest in (\ref{SafeMerging}).
Then by setting
\begin{align} \label{b1}
    b_2(\textbf{x}_i(t),\textbf{x}_{i_c}(t))&=z_{i,i_c}(t)-\varphi v_{i}(t)-\delta \nonumber\\
    &=x_{i_c}(t)-x_i(t)-\Phi(x_i(t)) v_i(t)-\delta,
\end{align}
proceeding as in the derivation of (\ref{CBF1}), we obtain:
\begin{align}\small \label{CBF2}
&\underbrace{v_{i_c}(t)-v_i(t)-\frac{\varphi}{L}v_i^2(t)}_{L_fb_2(\textbf{x}_i(t),\textbf{x}_{i_c}(t))}+\underbrace{-\varphi \frac{x_i(t)}{L}}_{L_gb_2(\textbf{x}_i(t))}u_i(t)+\nonumber \\ &\underbrace{k_2(z_{i,i_c}(t)-\varphi  \frac{x_i(t)}{L} v_i(t)-\delta)}_{\gamma_2(b_2(\textbf{x}_i(t),\textbf{x}_{i_c}(t)))} \geq 0.
\end{align}


The speed constraints in \eqref{VehicleConstraints1} are also easily transformed into CBF constraints using (\ref{CBF general constraint}) by defining 
\begin{equation} \label{b3}
b_3(\textbf{x}_i(t))=v_{\max}-v_i(t),
\end{equation}
\begin{equation}\label{b4}
b_4(\textbf{x}_i(t))=v_i(t)-v_{\min}.
\end{equation}
This yields:
\begin{equation} \label{CBF3}
\underbrace{-1}_{L_gb_3(\textbf{x}_i(t))}u_i(t)+\underbrace{k_3(v_{\max}-v_i(t))}_{\gamma_3(b_3(\textbf{x}_i(t)))} \geq 0
\end{equation}
\begin{equation}\label{CBF4}
    \underbrace{1}_{L_gb_4(\textbf{x}_i(t))}u_i(t)+\underbrace{k_4(v_i(t)-v_{\min})}_{\gamma_4(b_4(\textbf{x}_i(t)))} \geq 0,
\end{equation}
for the maximum and minimum velocity constraints, respectively.

\textbf{Inclusion of soft constraints in (\ref{eqn:energyobja}).} As a last step in the OCBF approach, we can exploit the versatility of the CBF method to include soft constraints expressed as terminal state costs in (\ref{eqn:energyobja}), e.g., the CAV achieving a desired terminal speed. This is accomplished
by using a Control Lyapunov Function (CLF) to track specific state variables in the reference trajectory if desired. A CLF $V(\textbf{x}_i(t))$ is similar to a CBF (see \cite{XIAO2021109592}). In our problem, letting $V(\textbf{x}_i(t))=(v_i(t)-v_{i}^\textrm{ref}(t))^2$ we can express the CLF constraint associated with tracking the CAV speed to a desired value $v_{i}^\textrm{ref}(t)$ (if one is provided) as follows:
\begin{equation}\label{CLF}
L_fV(\textbf{x}_i(t))+L_gV(\textbf{x}_i(t))u_i(t)+\epsilon V(\textbf{x}_i(t))\leq e_i(t),
\end{equation}
where $\epsilon >0 $ and $e_i(t)$ makes this a soft constraint.

Now that all the original problem constraints have been transformed into CBF constraints, we can formulate the OCBF problem as follows:
\begin{equation}\label{QP-OCBF}\small
\min_{u_i(t),e_i(t)}J_i(u_i(t),e_i(t)):=\int_{t_i^0}^{t_i^f}\big[\frac{1}{2}(u_i(t)-u_{i}^{\textrm{ref}}(t))^2+\lambda e^2_i(t)\big]dt
\end{equation}
subject to vehicle dynamics (\ref{VehicleDynamics}), the CBF constraints (\ref{CBF1}), (\ref{CBF2}), \eqref{CBF3}, \eqref{CBF4}, the control constraint \eqref{VehicleConstraints2}, and CLF constraint (\ref{CLF}).
Note that this is a decentralized optimization problem, as it only requires information sharing with a small number of \enquote{neighbor} CAVs, i.e. CAV $i_p$ and $i_c$ (if they exist). We denote this set of CAV neighbors by $\mathcal{R}_i(t)$ at time $t$:
\begin{equation} \label{eq:neighborset}
    \mathcal{R}_i(t)=\{i_p(t),i_c(t)\}.
\end{equation}
Note that $\mathcal{R}_i(t)$ in general can change over time, i.e., $i_c(t)$ changes when dynamic \enquote{resequencing} (discussed in \cite{2020Weidynreseq}) is carried out and $i_p(t)$ changes in the case of lane changing maneuvers. It is worth mentioning that in the single lane merging example in Fig. \ref{fig:merging} $i_p$ cannot change.

A common way to solve this dynamic optimization problem is to discretize $[t_i^0,t_i^f]$ into intervals $[t_i^0,t_i^0+\Delta),...,[t_i^0+k\Delta,t_i^0+(k+1)\Delta),...$ 
with equal length $\Delta$ and solving (\ref{QP-OCBF}) over each time interval. The decision variables $u_{i,k}=u_i(t_{i,k})$ and $e_{i,k}=e_i(t_{i,k})$ are assumed to be constant on each interval and can be easily calculated at time $t_{i,k}=t_i^0+k\Delta$ through solving a QP at each time step:
\begin{align} \label{QP}
\min_{u_{i,k},e_{i,k}}&[ \frac{1}{2}(u_{i,k}-u_{i}^{\textrm{ref}}(t_{i,k}))^2+\lambda e_{i,k}^{2}],
\end{align} 
subject to the CBF constraints (\ref{CBF1}), (\ref{CBF2}), \eqref{CBF3}, \eqref{CBF4}, and control input bounds \eqref{VehicleConstraints2} and CLF constraint (\ref{CLF}) where all constraints are linear in the decision variables. We refer to this as the \emph{time-driven} approach, which is fast and can be readily used in real time. 

The main problem with this approach is that a QP may become infeasible at any time instant because the decision variable $u_{i,k}$ is held constant over a given time period $\Delta$. Since this is externally defined, there is no guarantee that it is small enough to ensure the forward invariance property of a CBF, thereby also failing to ensure the satisfaction of the safety constraints. In other words, in this time-driven approach, there is a critical (and often restrictive) assumption that the control update rate is high enough to avoid such a problem. There are several additional issues worth mentioning: $(i)$ imposing a high update rate makes the solution of multiple QPs inefficient since it increases the computational burden, $(ii)$ using a common update rate across all CAVs renders their synchronization difficult, and $(iii)$ the predictability of a time-driven communication mechanism across CAVs makes the whole system susceptible to malicious attacks.  
As we will show next, the two event-driven solutions proposed in this paper alleviate these problems by eliminating the need to select a time step $\Delta$.


\section{EVENT-DRIVEN SOLUTIONS}
\label{sec:event-triggered}
There are several possible event-driven mechanisms one can adopt to invoke the solution of the QPs in (\ref{QP}) subject to the CBF constraints (\ref{CBF1}), (\ref{CBF2}), \eqref{CBF3}, \eqref{CBF4} along with control input bounds \eqref{VehicleConstraints2}. One approach is to adopt an \emph{event-triggering} scheme
such that we only need to solve a QP (with its associated CBF constraints) when one of two possible events (as defined next) is detected. We will show that this provides a guarantee for the satisfaction of the safety constraints which cannot be offered by the time-driven approach described earlier. The key idea is to ensure that the safety constraints are satisfied while the state remains within some bounds and define events which coincide with the state reaching these bounds, at which point the next instance of the QP in (\ref{QP}) is triggered.
Another idea is to create a \emph{self-triggering} framework with a minimum inter-event time guarantee by 
predicting at $t_{i,k}$ the first time instant that any of the CBF constraints in the QP problem (\ref{QP}) is subsequently violated. We then select that as the next
time instant $t_{i,k+1}$ when CAV $i$ communicates with the coordinator and updates the control.

\subsection{Event-triggered Control}

Let $t_{i,k}$, $k=1,2,...$, be the time instants when the QP in (\ref{QP}) is solved by CAV $i$. Our goal is to guarantee that the state trajectory does not violate any safety constraints within any time interval $[t_{i,k},t_{i,k+1})$ where $t_{i,k+1}$ is the next time instant when the QP is solved. 
Define a subset of the state space of CAV $l$ at time $t_{i,k}$ such that:
\begin{equation} \label{bound}
\textbf{x}_i(t_{i,k})-\textbf{s}_i \leq \textbf{x}_i(t) \leq \textbf{x}_i(t_{i,k})+\textbf{s}_i,
\end{equation}
where $\textbf{s}_i =\left[s_{i_x} \ \ s_{i_v} \right]^T \in \mathbb{R}_{>0}^2$ is a parameter vector whose choice will be discussed later. Intuitively, this choice reflects a trade-off between \emph{computational efficiency} (when the $\textbf{s}_i$ values are large and there are fewer instances of QPs to be solved) and \emph{conservativeness} (when the values are small). We denote the set of states of CAV $i$ that satisfy \eqref{bound} at time $t_{i,k}$ by 
\begin{equation} \label{event bound}
S_i(t_{i,k}) = \Bigl\{ \textbf{y}_i \in \textbf{X}: ~\textbf{x}_i(t_{i,k})-\textbf{s}_i \leq \textbf{y}_i \leq \textbf{x}_i(t_{i,k})+\textbf{s}_i\Bigl\}.
\end{equation} 
In addition, let $C_{i,1}$ be the feasible set of our original constraints \eqref{Safety}, \eqref{SafeMerging} and \eqref{VehicleConstraints1} defined as
\begin{equation} \label{event:ci1}
    C_{i,1}:=\Bigl\{  \mathbf{x}_i\in \mathbf{X}: ~b_q(\mathbf{x}_i)\geq 0, \ q \in \lbrace 1,2,3,4 \rbrace \Bigl\}.
\end{equation}

Next, we seek a bound and a control law that satisfies the safety constraints within this bound. This can be accomplished by considering the minimum value of each component in \eqref{CBF general constraint} for every $q \in \lbrace 1,2,3,4 \rbrace $  as shown next.

Let us start with the first of the three terms in \eqref{CBF general constraint}, $L_fb_q(\textbf{x}(t))$. Observing that not all state variables are generally involved in a constraint $b_q(\textbf{x}(t)) \ge 0$, we can rewrite this term as 
$L_fb_q(\textbf{y}_i(t),\textbf{y}_r(t))$ with $\mathbf{y}_i(t)$ as in (\ref{event bound}) and where $r$ stands for \enquote{relevant} CAVs affecting the specific constraint of $i$, i.e., $r \in \mathcal{R}_i(t)$ in (\ref{eq:neighborset}). 
Let $b^{\min}_{q,f_i}(t_{i,k})$ be the minimum possible value of the term $L_fb_q(\textbf{y}_i(t),\textbf{y}_r(t))$ 
over the time interval $[t_{i,k},t_{i,k+1})$ for each $q= \lbrace 1,2,3,4 \rbrace $ over the set $\Bar{S_i}({t_{i,k}}) \cap \Bar{S_r}({t_{i,k}})$:

\begin{equation}\label{minfi}
b^{\min}_{q,f_i}(t_{i,k})=\displaystyle\min_{\textbf{y}_i \in \Bar{S}_i({t_{i,k}}) \atop \textbf{y}_r \in \Bar{S}_r({t_{i,k}})}L_fb_q(\textbf{y}_i(t),\textbf{y}_r(t)),
\end{equation} 
where $\Bar{S}_i({t_{i,k}})$ is defined as follows:
\begin{equation}
    \Bar{S}_i({t_{i,k}}):=\lbrace\mathbf{y}_i \in C_{i,1} \cap S_i(t_{i,k}) \rbrace
\end{equation}

Similarly, we can define the minimum value of the third term in \eqref{CBF general constraint}:
\begin{equation}\label{mingammai}
b^{\min}_{\gamma_q}(t_{i,k})=\displaystyle\min_{\textbf{y}_i \in \Bar{S}_i({t_{i,k}}) \atop \textbf{y}_r \in \Bar{S}_r({t_{i,k}})} \gamma_q(\textbf{y}_i(t),\textbf{y}_r(t)).
\end{equation}

For the second term in \eqref{CBF general constraint}, note that $L_gb_q(\mathbf{x}_i)$ is a constant for $ q=\{1,3,4\} $, as seen in 
(\ref{CBF1}), (\ref{CBF3}) and (\ref{CBF4}),
therefore there is no need for any minimization. However, $L_gb_2(\mathbf{x}_i)=-\varphi \frac{x_i(t)}{L}$ in (\ref{CBF2}) is state-dependent and needs to be considered for the minimization. Since $x_i(t) \ge 0$, note that $L_gb_2(\mathbf{x}_i)$ is always negative, therefore, we can determine the limit value $b^{\min}_{2,g_i}(t_{i,k}) \in \mathbb{R}, $ as follows:
\begin{eqnarray}\label{mingi} \small
b^{\min}_{2,g_i}(t_{i,k})=\begin{cases}
\displaystyle\min_{\textbf{y}_i \in \Bar{S}_i({t_{i,k}}) \atop \textbf{y}_r \in \Bar{S}_r({t_{i,k}})}L_gb_2(\textbf{x}_i(t)), \  \textnormal{if}\  u_{i,k} \geq 0\\
\\
\displaystyle\max_{\textbf{y}_i \in \Bar{S}_i({t_{i,k}}) \atop \textbf{y}_r \in \Bar{S}_r({t_{i,k}})}L_gb_2(\textbf{x}_i(t)), \ \ \  \textnormal{otherwise},
\end{cases}
\end{eqnarray}
where the sign of $u_{i,k},  \ i \in F(t_{i,k})$ can be determined by simply solving the CBF-based QP  \eqref{QP} at time $t_{i,k}$.

Thus, the condition that can guarantee the satisfaction of \eqref{CBF1}, \eqref{CBF2} and \eqref{CBF3}, \eqref{CBF4} in the time interval $\left[t_{i,k},t_{i,k+1}\right)$ is given by
\begin{equation} \label{minCBF}
b^{\min}_{q,f_i}(t_{i,k})+b^{\min}_{q,g_i}(t_{i,k})u_{i,k}+b^{\min}_{\gamma_q}(t_{i,k})\geq 0,
\end{equation}
for $q=\lbrace1,2,3,4\rbrace$. In order to apply this condition to the QP \eqref{QP}, we just replace \eqref{CBF general constraint} by \eqref{minCBF} as follows:
\begin{align} \label{eq:QPtk}
\min_{u_{i,k},e_{i,k}}& \Bigl[ \frac{1}{2}(u_{i,k}-u_i^{\textmd{ref}}(t_{i,k}))^2+\lambda e_{i,k}^{2}\Bigl]\nonumber\\
 &\textnormal{s.t.} \ \  \eqref{CLF},\eqref{minCBF},\eqref{VehicleConstraints2}
\end{align} 
It is important to note that each instance of the QP \eqref{eq:QPtk} is now triggered by one of the following two events where $k =1,2,\ldots$ is a local event (rather than time step) counter:

\begin{itemize}
    \item \textbf{Event 1:} the state of CAV $i$ reaches the boundary of $S_i(t_{i,k-1})$.
    \item \textbf{Event 2:} the state of CAV $r \in \mathcal{R}_i(t_{i,k-1})$ reaches the boundary of $S_{r}(t_{i,k-1})$, if $\mathcal{R}_i(t_{i,k-1})$ is nonempty. In this case either $r=i_p$ or $r=i_c$ (e.g., in the merging problem $i_c=i-1 \neq i_p$ if such a CAV exists). Thus, Event 2 is further identified by the CAV which triggers it and denoted accordingly by \textbf{Event 2}($r$), $r \in \mathcal{R}_i(t_{i,k-1})$.
\end{itemize}

As a result, $t_{i,k},k=1,2,...$ is unknown in advance but can be determined by CAV $i$ through:
\begin{align} \label{events}
t_{i,k}=\min \Big\{ t>t_{i,k-1}:\vert\textbf{x}_i(t)-\textbf{x}_i(t_{i,k-1})\vert=\textbf{s}_i \\ \nonumber
\text{or} \ \ \vert\textbf{x}_{i_p}(t)-\textbf{x}_{i_p}(t_{i,k-1})\vert=\textbf{s}_{i_p} \\ \nonumber 
\text{or} \ \ \vert\textbf{x}_{i_c}(t)-\textbf{x}_{i_c}(t_{i,k-1})\vert=\textbf{s}_{i_c}\Big\},
\end{align}
where $t_{i,0}=t_i^0$. 
Note that $k$ is a \emph{local} event counter for each $i$ so, strictly speaking, we should use $k_i$. Instead, the index $k$ can be dropped and we can write $\textbf{x}_i(t_{i,\textrm{last}})$ rather than $\textbf{x}_i(t_{i,k-1})$. However, when there is no ambiguity, we will simply write $\textbf{x}_i(t_{i})$ to indicate that $t_i$ is the ``last event'' occurring at $i$.


The definition above is based on events which directly affect CAV $i$ (leading to $i$ solving a QP)  whether they are triggered by $i$ or $r \neq i$. Alternatively, we may think of any CAV $i$ as generating Event 1 leading to a new QP solution by CAV $i$ itself and Event 2($i$) which affects some $j \in \{l |i \in \mathcal{R}_l(t)\}$, i.e., $i$ is relevant to some $j \neq i$. In this case, a violation of the bound of $S_i(t_{i,k-1})$ or $S_i(t_{j,k-1})$ by the evolving state of CAV $i$ triggers events relevant to CAV $i$ or $j$, respectively.

Events 1,2($r$) can be detected through the dynamics in \eqref{VehicleDynamics} or from on-board state measurements, if available, along with state information from relevant other CAVs (e.g., CAVs $i_p$ and $i_c$ in Fig. \ref{fig:merging}) through the coordinator. Finally, note that because of the Lipschitz continuity of the dynamics in \eqref{VehicleDynamics} and the fact that the control is constant within an inter-event interval, Zeno behavior does not occur in this framework.


The following theorem formalizes our analysis by showing that if new constraints of the general form \eqref{minCBF} hold, then our original CBF constraints \eqref{CBF1}, \eqref{CBF2} and \eqref{CBF3}, \eqref{CBF4} also hold. 

\begin{theorem}\label{as:1} Given a CBF $b_q(\mathbf{x(t)})$ with relative degree one, let $t_{i,k}$, $k=1,2,\ldots$ be determined by \eqref{events} with $t_{i,0}=t_i^0$ and $b^{\min}_{q,f_i}(t_{i,k})$, $b^{\min}_{\gamma_q}(t_{i,k})$, $b^{\min}_{q,g_i}(t_{i,k})$  for $q=\{1,2,3,4\}$ obtained through \eqref{minfi}, \eqref{mingammai}, and \eqref{mingi}. Then, any control input $u_{i,k}$ that satisfies \eqref{minCBF} for all $q \in \lbrace 1,2,3,4 \rbrace$  within the time interval $[t_{i,k},t_{i,k+1})$ renders the set $C_{i,1}$ forward invariant for the dynamic system defined in (\ref{VehicleDynamics}).
\end{theorem}
\begin{proof}
The proof follows along similar lines as Theorem 2 in \cite{Xiao2021EventTriggeredSC}.
By \eqref{event bound}, we can write:
\begin{equation}
\textbf{y}_i(t) \in S_i(t_{i,k}), ~\textbf{y}_{r}(t) \in S_{r}(t_{i,k}), ~\textbf{y}_i(t) \in C_{i,1}
\end{equation}
for all $t \in [t_{i,k},t_{i,k+1}),  k=1,2,...$. 
\begin{equation}
L_fb_q(\textbf{x}_i(t)) \geq b^{\min}_{q,f_i}(t_k),
\end{equation}
\begin{equation}
\gamma_q(\textbf{x}_i(t)) \geq b^{\min}_{\gamma_q}(t_k),
\end{equation}
\begin{equation}
L_gb_q(\textbf{x}_i(t))u_i(t_k) \geq b^{\min}_{q,g_i}(t_k)u_i(t_k),
\end{equation}
for $q \in \lbrace1,2,3,4\rbrace$, by adding these inequalities which have the same direction it follows that 
\begin{align}
L_fb_q(\textbf{x}_i(t))+L_gb_q(\textbf{x}_i(t))u_i(t_k)+\gamma_q(\textbf{x}_i(t))\\ \nonumber 
\geq b^{\min}_{q,f_i}(t_k)+b^{\min}_{q,g_i}(t_k)u_i(t_k) +b^{\min}_{\gamma_q}(t_k) \geq 0. 
\end{align}
i.e., \eqref{CBF general constraint} is satisfied.
By Theorem 1 of \cite{XIAO2021109592} applied to \eqref{CBF general constraint}, if $x_i(0) \in
C_{i,1}$, then any Lipschitz continuous controller $u_i(t)$
that satisfies \eqref{CBF general constraint} $\forall t \geq 0$ renders $C_{i,1}$ forward
invariant for system \eqref{VehicleDynamics}. Therefore, $C_{i,1}$ is forward invariant for the
dynamic system defined in (\ref{VehicleDynamics}).
\end{proof}

\textbf{Remark} 1:
Expressing \eqref{minCBF} in terms of the minimum value of each component separately may become overly conservative if each minimum value corresponds to different points in the decision variable space. Therefore, an alternative approach is to calculate the minimum value of the whole term. 

\textbf{Selection of parameters $\mathbf{s}_i$.} The importance of properly selecting the parameters $\mathbf{s}_i$ is twofold. First, it is necessary to choose them such that all events are observed, i.e., given the sensing capabilities and limitations of a CAV $i$, the value of $\mathbf{s}_i$ must be large enough to ensure that no events will go undetected. 
In particular, the variation of the states of CAV $i$ within the sensor sampling time must not be greater than bounds $\mathbf{s}_i$.
Therefore, letting $T_s$ be a given sensor sampling time, the maximum state (position and speed) variation during this sampling time must satisfy:
\begin{equation}
    \begin{aligned} \label{min_s_x}
    x_i(t+T_s)-x_i(t) \leq v_{\max}T_s
    \end{aligned}
\end{equation}
\begin{equation}\label{min_s_v}
    \begin{aligned}
    v_i(t+T_s)-v_t(t) \leq \max(u_{\max}T_s,|u_{\min}|T_s)
    \end{aligned}
\end{equation}
where $v_{\max}$, $u_{\max}$, and $u_{\min}$ are given CAV $i$ specifications. Therefore, we need to pick lower bounds given by the maximum state variations in \eqref{min_s_x} and \eqref{min_s_v} as follows:
\begin{equation}
   \mathbf{s}_i=\left[
\begin{array}{cc}
     s_{ix}  \\
     s_{iv}
\end{array}\right]\geq \left[\begin{array}
[c]{c}%
v_{\max}T_s\\
\max(u_{\max}T_s,|u_{\min}|T_s)
\end{array}
\right].
\end{equation}

Second, the choice of $\mathbf{s}_i$ captures the trade-off between computational cost and conservativeness: the larger the value of each component of $\mathbf{s}_i$ is, the smaller the number of events that trigger instances of the QPs becomes, thus reducing the total computational cost. At the same time, the control law must satisfy the safety constraints over a longer time interval as we take the minimum values in \eqref{minfi}-\eqref{mingi}, hence rendering the approach more conservative.



{\bf Communication Scheme}. As mentioned earlier, a coordinator is responsible for exchanging information among CAVs (but does not exert any control). To accommodate event-triggered communication, the coordinator table in Fig. \ref{fig:merging} is extended as shown in Table \ref{Table c} so that it includes ``relevant CAV info'' data for each CAV $i$. In particular, in addition to the states of CAV $i$ in column 2, denoted by $\textbf{x}_i(t_{i})$, the states of CAVs $r \in \mathcal{R}_i(t_i)$ are included, denoted by $\textbf{x}_r(t_{i})$, in column 3, where CAV $r$ affects the constraints of CAV $i$, i.e., $r \in \mathcal{R}_i(t_i)$.  In an event-driven scheme, 
frequent communication is generally not needed, since it occurs only when an event is triggered.  CAV $i$ updates its state in the coordinator table and re-solves a QP in two cases depending on which event occurs: 

$(i)$ Event 1 triggered by $i$. The first step is state synchronization: CAV $i$ requests current states from all relevant CAVs and the coordinator updates these (column 3), as well as the state $\textbf{x}_i(t_i)$ (column 2). 
CAV $i$ then solves its QP while the coordinator notifies all CAVs $r \in \mathcal{R}_i(t_i)$ of the new CAV $i$ state 
so they can update their respective boundary set $S_r(t_{i})$. This may trigger an Event 2($r$) to occur at some future event time as in \eqref{events}; such an event cannot be triggered instantaneously, as it takes some finite time for a bound in $S_r(t_{i})$ to be reached because of Lipschitz continuity in the dynamics. In addition, the coordinator notifies all CAVs $j$ such that $i \in \mathcal{R}_j(t_i)$  (i.e. $i$ is relevant to $j$) so that they can update their bounds $S_j(t_i)$ respectively.

$(ii)$ Event 2($r$) is triggered by $r \in \mathcal{R}_i(t_i)$. When CAV $r$ reaches the boundary set $S_r(t_{i})$ it notifies the coordinator to update its state (column 2). The coordinator passes on this information to all CAVs $j$ where $r \in \mathcal{R}_j(t_i)$, which includes $i$ since $r \in \mathcal{R}_i(t_i)$, and the corresponding state of $r$ is updated (column 3). Then, CAV $i$ re-solves its QP and the coordinator updates $t_i$ to the current time and the state $\textbf{x}_i(t_i)$ (column 2) and the state $\textbf{x}_r(t_i)$(column 3). The rest of the process is the same as in case $(i)$. 

Note that any update in CAV $i$'s state due to a triggered event can immediately affect only CAVs $l>i$ such that $i$ is relevant to $l$. If an \enquote{event chain} ensues, the number of events is bounded by $N(t_i)$. 

\textbf{Remark} 2:
It is possible to simplify the communication scheme by assuming that each CAV can measure (through local sensors) the state of its relevant CAVs (i.e. in the case of CAV $i$, the states of the CAV  $r \in \mathcal{R}_i(t_i)$ ). Thus, CAVs can check for violations not only in their own state boundaries $S_i(t_i)$ but also in their relevant CAV state boundaries, $S_r(t_i)$. The same applies to the case where CAVs have a direct vehicle-to-vehicle (V2V) communication capability.


\begin{table}
        \begin{center}
        \begin{tabular}{|c|c|c|c|}
        \hline
        \multicolumn{4}{|c|}{Extended Coordinator Table} \\
        \hline
        Index & CAV Info & Relevant CAV info & Lane\\
        \hline
        0 & $\boldsymbol{x}_\textnormal{0}(t_{0})$ & - & Main\\
        \hline
         1 & $\textbf{x}_\textnormal{1}(t_{1})$ & - & Main\\
        \hline
        2 & $\textbf{x}_\textnormal{2}(t_{2})$ & $\textbf{x}_\textnormal{1}(t_{2})$ & Merging\\
        \hline
        3 & $\textbf{x}_\textnormal{3}(t_{3})$ & $\textbf{x}_\textnormal{1}(t_{3})$ ,$\textbf{x}_\textnormal{2}(t_{3})$ & Merging\\
        \hline
        4 & $\textbf{x}_\textnormal{4}(t_{4})$ & $\textbf{x}_\textnormal{1}(t_{4})$ ,$\textbf{x}_3(t_{4})$ & Main\\
        \hline
        5 & $\textbf{x}_\textnormal{5}(t_{5})$ & $\textbf{x}_\textnormal{4}(t_{5})$ & Main\\
        \hline
        \end{tabular}
        \caption{Extended coordinator table from Fig \ref{fig:merging}\\
        for event triggered control. }
        \label{Table c}
        \end{center}
\end{table}
\subsection{Self-Triggered Control}
As an alternative to event-triggered control, a self-triggered asynchronous control scheme can be used where each CAV $i$ communicates with the coordinator at specified time instants $\{t_{i,k}\},{k\in \mathbb{Z}^+}$. At each such instant $t_{i,k}$, CAV $i$ uploads its own state information $\textbf{x}_i(t_{i,k})$, the calculated control input $u_i(t_{i,k})$ that is going to be applied over the time interval $[t_{i,k},t_{i,k+1})$, and the next time when CAV $i$ will communicate with the coordinator and resolve its QP, denoted as $t_{i,{\textmd{next}}}$. The data stored at the coordinator for all vehicles are shown in Table \ref{tab:1}. We denote the most recent stored information of the $i$-th CAV at the coordinator as $\mathcal{I}_i=[t_{i,{\textmd{last}}},t_{i,{\textmd{next}}},x_i(t_{i,{\textmd{last}}}),v_i(t_{i,{\textmd{last}}}),u_i(t_{i,{\textmd{last}}})]
$. 
\begin{table}
\begin{center}
\begin{tabular}{|l|l|}
\hline
$t_{i,{\textmd{last}}}$ & Last time CAV $i$ communicated.\\
$t_{i,{\textmd{next}}}$ & Next time CAV $i$ will communicate.\\
$x_i(t_{i,{\textmd{last}}})$ & Last updated position of CAV $i$.\\
$v_i(t_{i,{\textmd{last}}})$ & Last updated velocity of CAV $i$.\\
 $u_i(t_{i,{\textmd{last}}})$ & Last control input of CAV $i$.\\
 \hline
\end{tabular}
\caption{Data stored on the coordinator for self-triggered control}\label{tab:1}
\end{center}
\end{table}

The goal is to develop a self-triggered asynchronous algorithm to determine the sequence of time instants  $t_{i,k}$ and the control input $u_i(t), t\in [t_{i,k},t_{i,k+1})$ for each CAV to solve the problem formed in \eqref{QP-OCBF}. Providing a lower bound for the inter-event time interval is an imperative feature in a self-triggered scheme. It is worth mentioning that since Zeno behavior never occurs under Lypschitz continuity, such a guarantee is not necessary for the event-triggered control algorithm. To provide such a guarantee for the generated time instants $t_{i,k}$, there should exist some $T_d>0$  such that  $|t_{i,k+1}-t_{i,k}|\geq T_d$. This is a design parameter which depends on the sensor sampling rate, as well as the clock of the on-board embedded system on each CAV. For the same reason, the time-instants  $t_{i,k}$ are calculated such that $(t_{i,k}~\textrm{mod}~T_d)=0$ where $\textrm{mod}$ denotes the modulo operator.
In contrast to the time-driven scheme with a fixed sampling time $\Delta$, each CAV $i \in F(t_{i,k})$ calculates the time instant $t_{i,k}$ in which the  QP problem must be solved in a self-triggered fashion. As in the event-triggered scheme, at each time instant $t_{i,k}$,  CAV $i$ solves its QP problem to obtain $u_i(t_{i,k})$. However, unlike the event-triggered scheme, CAV $i$ also calculates the next time instant $t_{i,k+1}$ at which it should resolve the QP problem. Note that similar to the time-driven scheme, the newly obtained control input, $u_i(t_{i,k})$ is held constant over the time interval  $[t_{i,k},t_{i,k+1})$ for CAV $i$. 


We address two problems in the following. First, it will be shown how a lower bound $T_d$ on the inter-event time interval can be ensured. Second, we will show how each CAV $i \in F(t_{i,k})$ specifies the time instants $t_{i,k}$.

\subsubsection{Minimum Inter-event Time, $T_d$}
In this subsection, it is shown how the CBF constraints \eqref{CBF1}, \eqref{CBF2}, \eqref{CBF3}, and \eqref{CBF4} for the CAV $i$ should be modified to ensure a minimum inter-event time $T_d$. This is achieved by adding extra positive terms 
to the right hand side of these constraints.
 First, consider the maximum speed CBF constraint \eqref{CBF3} to be satisfied when solving the QP problem at $t_{i,k}$ with feasible solution $u_i(t_{i,k})$. Thus, we have:
\begin{align} \label{cbf_i1}
\mathcal{C}_{i,1}(t_{i,k},u_i(t_{i,k}))&:=-u_i(t_{i,k})+k_3b_3(\textbf{x}_i(t_{i,k})) \geq 0.
\end{align}
However, the CBF constraint should be satisfied over the entire time interval $[t_{i,k},t_{i,k}+T_d]$ to ensure the minimum inter-event time. Therefore, for all $t\in [t_{i,k},t_{i,k}+T_d]$:
\begin{align} 
\mathcal{C}_{i,1}(t,u_i(t_{i,k}))&=-u_i(t_{i,k})+k_3b_3(\textbf{x}_i(t)) \geq 0.\label{eq_c1} 
\end{align}
By defining $\tau=t-t_{i,k}$ as the elapsed time after $t_{i,k}$, and recalling that the acceleration is kept constant over the inter-event time, we can derive an expression for the velocity $v_i(t)$ as follows:
\begin{equation} \label{vi}
v_i(\tau)=v_i(t_{i,k})+u_i(t_{i,k})\tau, \ \ \tau \in [0,T_d].
\end{equation}
Now by using \eqref{CBF3}, \eqref{cbf_i1}, and, \eqref{vi} we can rewrite \eqref{eq_c1} as follows:
\begin{align} \label{cbf_i1-eq_c1}
\mathcal{C}_{i,1}(t,u_i(t_{i,k}))&=\mathcal{C}_{i,1}(t_{i,k},u_i(t_{i,k}))-k_3u_i(t_{i,k})\tau  \ \ \tau \in [0,T_d],
\end{align}
In what follows, we show that if $\mathcal{C}_{i,1}(t_{i,k},u_i(t_{i,k}))\geq \sigma_{i,1}(T_d)$ holds, then:
\begin{equation} \label{temp 1}
    \mathcal{C}_{i,1}(t,u_i(t_{i,k}))\geq 0, \ \ \ \forall t \in [t_{i,k},t_{i,k}+T_d],
\end{equation}
where $\sigma_{i,1}(T_d):=k_3u_MT_d$ and $u_M=\max(|u_{\min}|,u_{\max}) >0$.  To prove \eqref{temp 1}, we can rewrite $\mathcal{C}_{i,1}(t_{i,k},u_i(t_{i,k}))\geq \sigma_{i,1}(T_d)$ as follows:
\begin{align} \label{inequality1}
 \mathcal{C}_{i,1}(t_{i,k},u_i(t_{i,k}))&- \mathcal{C}_{i,1}(t,u_i(t_{i,k})) \nonumber \\&+   \mathcal{C}_{i,1}(t,u_i(t_{i,k}))\geq \sigma_{i,1}(T_d)
\end{align}
By combining \eqref{inequality1} with \eqref{cbf_i1-eq_c1}, for all $t \in [t_{i,k},t_{i,k}+T_d]$ and $\tau \in [0,T_d]$ we have:
\begin{align}
    \mathcal{C}_{i,1}(t,u_i(t_{i,k})) \geq \sigma_{i,1}(T_d)&-k_3u_i(t_{i,k})\tau \geq 0,
\end{align}
where the non-negativity of the inequality follows from the definition of $\sigma_{i,1}(t_{i,k},T_d)=k_3u_MT_d$, i.e., $k_3u_MT_d-k_3u_i(t_{i,k})\tau \geq 0$ for $\tau \in [0,T_d]$. Hence, in order to ensure the minimum inter-event interval $T_d$, the CBF constraint \eqref{CBF3} should be modified to:
\begin{align}\label{cbf_modified_1}
   \mathcal{C}_{i,1}(t,u_i(t))\geq \sigma_{i,1}(T_d).
\end{align}
Following a similar derivation of the modified CBF constraint for the minimum speed \eqref{CBF3}, it follows that \eqref{CBF4} should be modified to:
\begin{align}\label{cbf_modified_2}
 \mathcal{C}_{i,2}(t,u_i(t)) \geq \sigma_{i,2}(T_d),
\end{align}
where 
\begin{align}
    \mathcal{C}_{i,2}(t,u_i(t))&:=u_i(t)+k_4b_4(\textbf{x}_i(t)) \nonumber \\ \sigma_{i,2}(T_d)&:=k_4u_{M}T_d.
\end{align}

Next, let us consider the safety CBF constraint \eqref{CBF1} to be satisfied when solving the QP problem at $t_{i,k}$ with a feasible solution $u_i(t_{i,k})$. It follows that 
\begin{align} \nonumber
  \mathcal{C}_{i,3}(t_{i,k},u_i(t_{i,k})):= & v_{i_p}(t_{i,k})-v_i(t_{i,k})-{\varphi} u_i(t_{i,k})\\ &+k_1b_1(\textbf{x}_i(t_{i,k}),\textbf{x}_{i_p}(t_{i,k}))\geq 0. \label{eq_c2}
\end{align}
Once again, we need to ensure that the CBF constraint is satisfied over the entire time interval $[t_{i,k},t_{i,k}+T_d]$ as follows:
\begin{align} \nonumber
     \mathcal{C}_{i,3}(t,u_i(t_{i,k})) &=v_{i_p}(t)-v_i(t)-{\varphi} u_i(t_{i,k})\\&+k_1b_1(\textbf{x}_i(t),\textbf{x}_{i_p}(t)) \geq 0, \ \ t \in [t_{i,k},t_{i,k}+T_d]. \label{eq_30}
\end{align}
For ease of notation, we use the following definitions:
    \begin{equation} \label{viip_tik}
     \Delta v_{i,i_p}(t_{i,k})=v_{i_p}(t_{i,k})-v_i(t_{i,k}),
    \end{equation}
    \begin{equation}\label{uiip_tik}
       \Delta u_{i,i_p}(t_{i,k})=u_{i_p}(t_{i,k})-u_i(t_{i,k}).
    \end{equation}
Similar to the procedure of deriving the lower bound for constraints \eqref{CBF3} and \eqref{CBF4}, by using \eqref{CBF1}, \eqref{viip_tik}, and \eqref{uiip_tik}, we rewrite \eqref{eq_30} as follows:
\begin{align} 
 \mathcal{C}_{i,3}(t,u_i(t_{i,k}))= & \mathcal{C}_{i,3}(t_{i,k},u_i(t_{i,k}))+\Delta u_{i,i_p}(t_{i,k})\tau \nonumber \\  &+ k_1 \bigl(  0.5\Delta u_{i,i_p}(t_{i,k})\tau^2 +\Delta v_{i,i_p}(t_{i,k})\tau \nonumber\\&-\varphi u_i(t_{i,k})\tau  \bigl) \geq 0,  \ \ \ \tau \in [0,T_d] \label{Ci3}.
\end{align}
To further ease up the notation, we define
\begin{equation} \label{M3}
    \mathcal{M}_{i,3}(t,t_{i,k},u_i(t_{i,k})):=  \mathcal{C}_{i,3}(t_{i,k},u_i(t_{i,k}))-\mathcal{C}_{i,3}(t,u_i(t_{i,k})),
\end{equation} 
which will be used later on.
Similarly, in the following we intend to show that if $\mathcal{C}_{i,3}(t_{i,k},u_i(t_{i,k}))\geq  \sigma_{i,3}(t_{i,k},T_d)$ holds, it follows:
\begin{equation} \label{temp 3}
   \mathcal{C}_{i,3}(t,u_i(t_{i,k}))\geq 0, \ \ \ t \in [t_{i,k},t_{i,k}+T_d],
\end{equation}
where
\begin{align}
    \sigma_{i,3}(t_{i,k},T_d) := &|u_{i_p}(t_{i,k})|+k_1 \bigl(0.5 T_d^2 (|u_{i_p}(t_{i,k})|+u_M) \nonumber \\
    &+(|v_{i,i_p}(t_{i,k})|+(1+\varphi)u_M)T_d\bigl).
\end{align}
To demonstrate \eqref{temp 3}, we follow the same procedure as before by starting with
\begin{equation} \label{temp 2}
   \mathcal{C}_{i,3}(t_{i,k},u_i(t_{i,k}))\geq  \sigma_{i,3}(t_{i,k},T_d),
\end{equation}
and then rewrite \eqref{temp 2} in the following form:
\begin{align} \label{inequality2}
\mathcal{C}_{i,3}(t_{i,k},u_i(t_{i,k}))&- \mathcal{C}_{i,3}(t,u_i(t_{i,k})) \nonumber \\&+  \mathcal{C}_{i,3}(t,u_i(t_{i,k}))\geq \sigma_{i,3}(t_{i,k},T_d)
\end{align}
Then, combining \eqref{M3} and \eqref{inequality2}, for $t \in [t_{i,k},t_{i,k}+T_d]$ follows that:
\begin{align}
    \mathcal{C}_{i,3}(t,u_i(t_{i,k})) &\geq \sigma_{i,3}(t_{i,k},T_d) - \mathcal{M}_{i,3}(t,t_{i,k},u_i(t_{i,k}))
\end{align}
where $\sigma_{i,3}(t_{i,k},T_d)$, i.e. the upper bound of $\mathcal{M}_{i,3}(t,t_{i,k},u_i(t_{i,k}))$, is chosen such that the left hand side of the inequality is always positive:
\begin{align}
  \sigma_{i,3}(t_{i,k},T_d) - \mathcal{M}_{i,3}(t,t_{i,k},u_i(t_{i,k})) \geq 0,
\end{align}
hence, by modifying the CBF constraint \eqref{CBF1} to:
\begin{align} \label{cbf_modified_3}
\mathcal{C}_{i,3}(t,u_i(t))\geq  \sigma_{i,3}(t,T_d),
\end{align}
one can enforce \eqref{eq_30}. 

Following a similar approach, to provide a minimum inter-event time $T_d$, the CBF constraint \eqref{CBF2} should be modified to, 
\begin{align} \label{cbf_modified_4}
  \mathcal{C}_{i,4}(t,u_i(t))\geq \sigma_{i,4}(t,T_d),
\end{align}
where
\begin{align*}
    \mathcal{C}_{i,4}(t,u_i(t))=&v_{i_c}(t)-v_i(t)-\frac{\varphi}{L}v_i^2(t)-\varphi \frac{x_i(t)}{L}u_i(t)+\nonumber \\ &k_2(b_2(\textbf{x}_i(t),\textbf{x}_{i_c}(t)),
\end{align*}
\begin{align}
    \sigma_{i,4}(t,T_d):=&0.5\frac{\varphi}{L}u^2_M T_d^3 + \nonumber \\
    +&k_4\Big( \frac{3\varphi}{2L}(u_M^2+|v_i(t)| u_M)+0.5 (|u_{i_c}(t)|+u_M) \Big)T^2_d \nonumber \\
    +&\Big(|u_{i_c}(t)|+ ( \frac{3\varphi}{L}|v_i(t)|+\frac{\varphi}{L}|x_i(t)|+1)u_M\nonumber \\
    +&|v_{i_c}(t)|+|v_i(t)|+\frac{\varphi}{L}v_i^2(t) \Big)k_4 T_d.
\end{align}


Finally, since the CLF constraint \eqref{CLF} is added optionally for an optimal trajectory, it can be relaxed in the presence of safety constraints and there is generally no need to ensure that it is satisfied during the whole time-interval $t \in [t_{i,k},t_{i,k}+T_d]$ with the same relaxation variable value $e_i(t_{i,k})$. Therefore, there is no need to modify it as was necessary for the CBF constraints. In conclusion, to ensure the minimum inter-event time $T_d$, at each time instant $t_{i,k}$, CAV $i$ needs to solve the following QP:
\begin{align}  \label{QP2}
 \min_{u_{i,k},e_{i,k}}~~\frac{1}{2}(u_{i,k}-u_i^{\textmd{ref}}(t_{i,k}))^2+\lambda e_{i,k}^2
\end{align}
subject to the modified CBF constraints \eqref{cbf_modified_1}, \eqref{cbf_modified_2}, \eqref{cbf_modified_3}, and \eqref{cbf_modified_4}, the control input bounds \eqref{VehicleConstraints2} and the CLF constraint \eqref{CLF}. In the next subsection, it will be shown how the time-instant $t_{i,k}$ should be obtained for CAV $i$.

\subsubsection{Self-Triggered Time Instant Calculation}

The key idea in the self-triggered framework is to predict the first time instant that any of the CBF constraints \eqref{CBF1}, \eqref{CBF2}, \eqref{CBF3} or \eqref{CBF4}, is violated and select that as the next time instant $t_{i,k+1}$. CAV $i$ then communicates with the coordinator and requests the necessary information to solve its next QP and obtain a new control input $u_i(t_{i,k+1})$ and update its stored data in the coordinator table. Note that it is not required to consider the modified CBF constraints \eqref{cbf_modified_1}, \eqref{cbf_modified_2}, \eqref{cbf_modified_3}, and \eqref{cbf_modified_4} here, since these are obtained purely for ensuring the minimum inter-event time $T_d$, while the original CBF constraints \eqref{CBF1}, \eqref{CBF2}, \eqref{CBF3}, and \eqref{CBF4} are sufficient for satisfying constraints 1, 2, and 3 (state limitation constraint) in problem \eqref{eqn:energyobja}.

For the speed constraint \eqref{CBF3},  it is clear that if $u_i(t_{i,k})\leq0$ (decelerating), then this constraint always holds, hence there is no need to check it. However, for $u_i(t_{i,k}) >0$ (accelerating), the constraint \eqref{CBF3} can be violated.  To calculate the time instant, $t^1_{i,k}$, when this occurs we need to solve the following equation:
\begin{equation} \label{temp 4}
    -u_i(t_{i,k})+k_3(v_{\max}-v_i(t))=0, \ \ \ \  t > t_{i,k}.
\end{equation}
Recalling that the acceleration is held constant in the inter-event time, \eqref{temp 4} can be rewritten as
\begin{equation}\label{temp 5}
    -u_i(t_{i,k})+k_3(v_{\max}-v_i(t_{i,k})-u_i(t_{i,k}) (t-t_{i,k}))=0
\end{equation} 
and its solution yields:
\begin{align*}
t^1_{i,k}=t_{i,k}+\frac{-u_i(t_{i,k})+k_3v_{\textmd{max}}-k_3v_i(t_{i,k})}{k_3u_i(t_{i,k})}.
\end{align*}
{Observe that at $t_{i,k}$, the QP in \eqref{QP2} is solved, therefore the constraint \eqref{cbf_modified_1} is satisfied at $t=t_{i,k}$ and we have $-u_i(t_{i,k})+k_3(v_{\textmd{max}}-v_i(t_{i,k}))  \geq \sigma_{i,1}(T_d) > 0$. It follows that $t^1_{i,k}\geq t_{i,k}+T_d$. }

For the second speed constraint \eqref{CBF4},  it is clear that if $u_i(t_{i,k})\geq0$ (accelerating), then this constraint is satisfied, hence there is no need to check it. However, for $u_i(t_{i,k}) <0$ (decelerating), the constraint \eqref{CBF4} can be violated. Similar to the previous case, we can solve the following equation for $t$ to obtain $t^2_{i,k}$ as the first time instant that constraint \eqref{CBF4} is violated:
\begin{equation}\label{temp 6}
    u_i(t_{i,k})+k_4(v_i(t_{i,k})+u_i(t_{i,k}) (t-t_{i,k})-v_{\textmd{min}}) = 0 \ \ \ \  t > t_{i,k}.
\end{equation}
Solving \eqref{temp 6} leads to
\begin{align*}
t^2_{i,k}=t_{i,k}+\frac{-u_i(t_{i,k})+k_4v_{\textmd{min}}-k_4v_i(t_{i,k})}{k_4u_i(t_{i,k})},
\end{align*}
and it can be shown, similar to the previous case, that $t^2_{i,k}\geq t_{i,k}+T_d$.

For the rear-end safety constraint \eqref{CBF1},  we need to find the first time instant $t>t_{i,k}$ such that $\mathcal{C}_{i,3}(t,u_i(t_{i,k}))=0$ in 
 \eqref{Ci3}. This leads to the following quadratic equation:
\begin{align*}
k_1&\big (0.5  \Delta u_{i,i_p}(t_{i,k})\big)\tau^2+\big ( \Delta u_{i,i_p}(t_{i,k})+k_1(\Delta v_{i,i_p}(t_{i,k})\\&-\varphi u_i(t_{i,k}) )\big )\tau+ \mathcal{C}_{i,3}(t_{i,k},u_i(t_{i,k}))=0.
\end{align*}
The least positive root of the above equation is denoted  as $\tau_{i,3}$ and we define $t^3_{i,k}=t_{i,k}+\tau_{i,3}$. { The case of having both roots negative corresponds to the constraint \eqref{CBF1} not being violated, hence $t^3_{i,k}=\infty$. } Moreover, due to the added term in \eqref{cbf_modified_3}, it follows that $t^3_{i,k}\geq t_{i,k}+T_d$.

Similarly for the safe merging constraint \eqref{CBF2},  the first time instant  $t>t_{i,k}$ such that $\mathcal{C}_{i,4}(t,u_i(t_{i,k}))=0$ can be obtained by solving the following cubic equation: 

\begin{align*}
&-k_4\frac{\varphi}{2L} u^2_i(t_{i,k})\tau^3+ \big (0.5  \Delta u_{i,i_c}(t_{i,k}) -k_4 \frac{3\varphi}{2L}u^2_i(t_{i,k})+
\\&-k_4\frac{3\varphi}{2L} v_{i}(t_{i,k}) u_i(t_{i,k}) \big ) \tau^2
+k_4\Big (\Delta u_{i,i_c}(t_{i,k})-\frac{3\varphi}{L}v_i(t_{i,k})u_i(t_{i,k})\\
&+  ( \Delta v_{i,i_c}(t_{i,k})-\frac{\varphi}{L} v^2_i(t_{i,k})-\frac{\varphi}{L} u_i(t_{i,k})x_{i}(t_{i,k})) \Big)\tau\\
 &+\mathcal{C}_{i,4}(t_{i,k},u_i(t_{i,k}))=0,
\end{align*}
where
\begin{align*}
    \Delta v_{i,i_c}(t_{i,k})=v_{i_c}(t_{i,k})-v_i(t_{i,k})\\
    \Delta u_{i,i_c}(t_{i,k})=u_{i_c}(t_{i,\textmd{k}})-u_i(t_{i,k}).
\end{align*}
The least positive root   is denoted  as $\tau_{i,4}$ and we define $t^4_{i,k}=t_{i,k}+\tau_{i,4}$. Moreover, due to solving QP in \eqref{QP2} subject to the modified CBF constraint derived in \eqref{cbf_modified_4}, it follows that $t^4_{i,k}\geq t_{i,k}+T_d$.  The case of having all roots negative corresponds to the constraint \eqref{CBF2} not being violated, hence $t^4_{i,k}=\infty$.

\subsubsection{Self-Triggered Scheme}

First, it should be noted that the time instants $t^q_{i,k}$, $q=1,\dots,4$  are obtained based on the safety constraints \eqref{Safety} and \eqref{SafeMerging}, as well as the vehicle state limitations \eqref{VehicleConstraints1}. However, this choice can compromise the optimal performance of CAVs in the CZ. In particular, it is possible that the acceleration of a CAV stays constant for a long period of time if there are no safety constraints or vehicle state limit violations, whereas, as shown in \cite{XIAO2021109592}, the optimal acceleration trajectory of the CAV in fact changes linearly. Therefore, in order to avoid this issue and minimize deviations from the optimal acceleration trajectory, one can impose a maximum allowable inter-event time, denoted by $T_{\max}$. To accomplish this, we can define
\begin{align} \label{eq:tmin}
    t^{\min}_{i,k}=\min \Bigl\{t^1_{i,k},t^2_{i,k},t^3_{i,k},t^4_{i,k},t_{i,k}+T_{\max} \Bigl\}.
\end{align}

The next update time instant for CAV $i$, i.e. $t_{i,k+1}=t_{i,\textmd{next}}$ should now be calculated. Towards this goal, consider the case where $t^{\min}_{i,k} \leq \min(t_{i_p,\textmd{next}},t_{i_c,\textmd{next}})$,
which corresponds to the next update time instant of CAV $i$ occurring before the next control update of the preceding vehicle $i_p$ or the conflict CAV $i_c$. Then, we can set $t_{i,k+1}=t_{i,\textmd{next}}=t^{\min}_{i,k}$ from (\ref{eq:tmin}).

The only remaining case is when $t^{\min}_{i,k}> \min(t_{i_p,\textmd{next}},t_{i_c,\textmd{next}})$, which corresponds to either CAV $i_p$  or $i_c$ updating its control input sooner than CAV $i$, hence CAV $i$ does not have access to their updated control input. Consequently, checking the constraints \eqref{CBF1} and \eqref{CBF2}  is no longer valid. In this case, $t_i^\textmd{next}= \min(t_{i_p,\textmd{next}},t_{i_c,\textmd{next}})+T_d$ which implies that CAV $i$'s next update time will be immediately  after the update time of  CAV $i_p$  or $i_c$ with a minimum inter-event time interval $T_d$. 

By setting $t_{r,\textrm{next}}^{\min}=\min(t_{{i_p},\textmd{next}},t_{i_c,\textmd{next}})$, we can summarize the selection of the next self-triggered time instant as follows:
\begin{align}\label{t_next}
    t_{i,{\textrm{next}}}=\left \{ \begin{array}{ll}
        t_{i,k}^{\min}, \ \ \ \ \ \ \ \ \ t_{i,k}^{\min}\leq t_{r,\textrm{next}}^{\min} \\
       t_{r,\textrm{next}}^{\min}+ T_d, \ \ \ \  \textrm{otherwise},
    \end{array} \right.
\end{align} 
Finally, in order to have $(t_{i,k}~\textrm{mod}~T_d)=0$, we set $t_{i,\textmd{next}}=\lfloor \frac{t_{i,\textmd{next}}}{T_d} \rfloor \times T_d$. 

It should be noted that the case of  $t_{i,{\textrm{next}}}=t_{i_c,{\textrm{next}}}$ or $t_{i,{\textrm{next}}}=t_{i_p,{\textrm{next}}}$  corresponds to having identical next update times for CAV $i$ and CAV $i_c$ or CAV $i_p$ so that they need to solve their QPs at the same time instant. However, in order for CAV $i$ to solve its QP at the time instant $t_{i,k+1}=t_{i,\textmd{next}}$, it requires the updated control input of CAV $i_c$ or CAV $i_p$, i.e. $u_{i_c}(t_{i,k+1})$ or $u_{i_p}(t_{i,k+1})$; this is practically not possible. In order to remedy this issue, whenever  $t_{i,{\textrm{next}}}=t_{i_c,{\textrm{next}}}$ or $t_{i,{\textrm{next}}}=t_{i_p,{\textrm{next}}}$, CAV $i$ solves its QP at $t_{i,k+1}$ by  using  $u_M$ instead of $u_{i_c}(t_{i,k+1})$ and  $u_{i_p}(t_{i,k+1})$ in \eqref{cbf_modified_3} and \eqref{cbf_modified_4}. This corresponds to considering the worst case in $\sigma_{i,3}(t,T_d)$ and $\sigma_{i,4}(t,T_d)$. Moreover, since calculating the next update time $t_{i,k+2}$ also depends on $u_{i_c}(t_{i,k+1})$ and  $u_{i_p}(t_{i,k+1})$, CAV $i$ in this case acts similar to the time-driven case with assigned  $t_{i,k+2}=t_{i,k+1}+T_d$. Then, at the next time instant $t_{i,k+2}$, CAV $i$ can obtain the updated control inputs of CAV $i_c$ and CAV $i_p$ from the coordinator and follows the proposed self-triggered scheme. 

{\bf Communication Scheme}.
In view of the constraints \eqref{CBF1} and \eqref{CBF2}, CAV $i$ requires knowledge of  $t_{i_p,last}, v_{i_p}(t_{i,k})$, \ $x_{i_p}(t_{i,k})$,  $t_{i_c,last}, v_{i_c}(t_{i,k})$, and $x_{i_c}(t_{i,k})$ at time instant $t_{i,k}$. Hence, at each time instant that it accesses the coordinator, it needs to download the recorded data of CAV $i_p$ and $i_c$. 
Then, the required updated information at  $t_{i,k}$ for CAV $i_p$ can be calculated as
\begin{equation}
    v_{i_p}(t_{i,k})= v_{i_p}(t_{i_p,\textmd{last}})+(t_{i,k}-t_{i_p,\textmd{last}})u_{i_p}(t_{i_p,\textmd{last}})
\end{equation}
\begin{align}
    x_{i_p}(t_{i,k})= x_{i_p}(t_{i_p,\textmd{last}})+&(t_{i,k}-t_{i_p,\textmd{last}})v_{i_p}(t_{i_p,\textmd{last}})\nonumber \\
    +&\frac{1}{2}(t_{i,k}-t_{i_p,\textmd{last}})^2u_{i_p}(t_{i_p,\textmd{last}})
\end{align}
with similar information calculated for CAV $i_c$. Note that the information for CAV $i_p$ may also be obtained from the on-board sensors at CAV $i$ if such are available. There are two key differences between the event-triggered and self-triggered approaches in the communication scheme as follows: (i). In the self triggered approach in addition to the states of the CAVs $\mathbf{x}_i(t_{i,last})$, control input $u_i(t_{i,last})$, current time instant $t_{i,last}$, and the next time instant of solving QP $t_{i,next}$ have to be shared. Whereas in the event-triggered only states of the CAVs and the states of the relevant CAVs at the time of the QP solving is needed. (ii). In this scheme, unlike the event-triggered scheme, the coordinator does not notify other relevant CAVs when a particular CAV solves QP and updates its data since the next QP solving time is known and stored in the coordinator table. For example, when CAV $i$ solves its QP there is no need for the CAVs $j$ where $i \in R_j(t_i)$ to be notified as they are already aware. Instead, the coordinator only receives and stores the current time instant, states, control input, and the next time instant of solving QP of the CAV $i$. Also upon download request from a particular CAV at the time of solving QP, access to the data of the relevant CAVs $r$ will be given to that particular CAV by coordinator. 


\section{SIMULATION RESULTS}
\label{sec:simulation}

All algorithms in this section have been implemented using \textsc{MATLAB}.
We used \textsc{quadprog} for solving QPs of the form
\eqref{QP}, \eqref{eq:QPtk} and \eqref{QP2}, \textsc{lingprog} for solving the linear programming in \eqref{minfi}, \eqref{mingammai} and \eqref{mingi}, \textsc{fmincon} for a nonlinear optimization problem arising when \eqref{minfi} and \eqref{mingammai} become nonlinear, and \textsc{ode45} to integrate the vehicle dynamics.

We have considered the merging problem shown in Fig. \ref{fig:merging} where CAVs are simulated according to Poisson arrival processes with an arrival rate which is fixed for the purpose of comparing the time-driven approach and the event-driven schemes (over different bound values in \eqref{events} for the event-triggered scheme and with different $T_{\max}$ for the self-triggered scheme). The initial speed $v_{i}(t_{i,0})$ is also randomly generated with a uniform distribution over $[15 \textnormal{m/s}, 20\textnormal{m/s}]$ at the origins $O$ and $O^{\prime}$, respectively. The
parameters for \eqref{QP-OCBF}, \eqref{eq:QPtk}, and \eqref{QP2}
 are: $L = 400\textnormal{m}, \varphi = 1.8\textnormal{s}, \delta = 0\textnormal{m}, u_{\max} = 4.905 \textnormal{m/s}^2, u_{\min} = -5.886\textnormal{m/s}^2, v_{\max} = 30\textnormal{m/s}, v_{\min} = 0\textnormal{m/s},  k_1=k_2=k_3=k_4=1,  \lambda= 10$ and $T_d=0.05$. The sensor sampling rate is $20$Hz, sufficiently high to avoid missing any triggering event as discussed earlier.
 The control update period for the time-driven control is $\Delta t=0.05$s. For the event-triggered scheme, we let the bounds $S=[s_x,s_v]$ be the same for the all CAVs in the network and vary them between the values of $\lbrace[0.5,1.5],[0.5,2],[0.5,2.5]\rbrace$. For the self-triggered scheme, we set $T_{\max} \in \{0.5,1,1.5,2\}$ to allow a comprehensive comparison. 
 
In our simulations, we included the computation of a more realistic energy consumption model \cite{kamal2012model} to supplement the simple surrogate $L_2$-norm ($u^2$) model in our analysis:
$f_{\textrm{v}}(t)=f_{\textrm{cruise}}(t)+f_{\textrm{accel}}(t)$ with
\begin{align*}
    f_{\textrm{cruise}}(t) &= \omega_0+\omega_1v_i(t)+\omega_2v^2_i(t)+\omega_3v^3_i(t),\\
    f_{\textrm{accel}}(t) &=(r_0+r_1v_i(t)+r_2v^2_i(t))u_i(t).
\end{align*}
where we used typical values for parameters $\omega_1,\omega_2,\omega_3,r_0,r_1$ and, $r_2$ as reported in \cite{kamal2012model}.

Our results from several simulations corresponding to the three different methods under the same conditions with different values for the relative weight of energy vs time are shown in Tables \ref{Table event} and \ref{Table self}: the time-driven method, the event-triggered scheme, and, the self-triggered scheme.
We observe that by using the event-triggered and self-triggered approaches we are able to significantly reduce the number of infeasible QP cases (up to $95\%$) compared to the time-driven approach. At the same time, the overall number of instances when a QP needs to be solved has also decreased up to $68\%$ and $80\%$ in the event-triggered and self-triggered approaches, respectively. 
Note that the large majority of infeasibilities is due to holding acceleration constant over an inappropriate sampling time, which can invalidate the forward invariance property of CBFs over the entire time interval. These infeasible cases were eliminated by the event-triggered and self-triggered schemes. However, another source of infeasibility is due to  conflicts that may arise between the CBF constraints and the control bounds in a QP. This cannot be remedied through the proposed event-triggered or self-triggered QPs; it can, however, be dealt with by the introduction of a sufficient condition that guarantees no such conflict, as described in \cite{XIAO2022inf}. 

In Tables \ref{Table event} and \ref{Table self}, we can also observe some loss of performance (i.e. average travel time increases hence road throughput decreases) in both approaches as the values of the bound parameters in the event triggered approach and $T_{\max}$ in the self triggering approach increases, hence increasing conservativeness. On the other hand, this decreases the computational load expressed in terms of the number of QPs that are solved in both methods, illustrating the trade-off discussed in previous sections. 
There is also an apparent discrepancy in the energy consumption results: when the $L_2$-norm of the control input is used as a simple metric for energy consumption, the values are higher under event-triggered and self triggered control, whereas the detailed fuel consumption model shows lower values compared to time-driven control. This is due to the fact that $u_i^2$ penalizes CAVs when they decelerate, whereas this is not actually the case under a realistic fuel consumption model. 

\begin{table*}\scriptsize
        \centering
        \begin{tabular}{|c|c|c|c|c|c|}
            \cline{1-6}
             &Item & \multicolumn{3}{|c|}{Event triggered} & Time driven\\
            \cline{2-6}
            & Bounds & $s_v=0.5, s_x=1.5$ & $s_v=0.5, s_x=2$ & $s_v=0.5, s_x=2.5$  & $\Delta t = 0.05$\\
        \hline  
        \multirow{4}{*}{\makecell{$\alpha=0.1$ }} & Ave. Travel time & 19.61 & 19.73 & 19.65  & 19.42\\
        \cline{2-6}
        & Ave. $\frac{1}{2} u^2$ & 4.45 & 4.81  & 5.16    & 3.18\\
        \cline{2-6}
        & Ave. Fuel consumption &31.77 & 31.51 & 31.04  & 31.61\\
        \cline{2-6}
        &Computation load (Num of QPs solved) & 50\% (17853)&  47\% (16778) & 34\% (12168) &  100\% (35443) \\
        \cline{2-6}
        & Num of infeasible cases   & \textcolor{blue}{42}  & \textcolor{blue}{42} & \textcolor{blue}{43} &  \textcolor{red}{315}\\
        \hline
        \multirow{4}{*}{\makecell{$\alpha=0.25$ }}  & Ave. Travel time & 15.82 & 15.88 & 15.95  & 15.44\\
        \cline{2-6}
        & Ave. $\frac{1}{2} u^2$& 13.93 & 14.06 & 14.25 &  13.34\\
        \cline{2-6}
        & Ave. Fuel consumption & 52.12 & 51.69  & 51.42 &  55.81 \\
        \cline{2-6}
        &Computation load (Num of QPs solved) &51\% (14465) & 51\% (14403) & 48\% (13707) &  100\% (28200)\\
        \cline{2-6}
        & Num of infeasible cases   & \textcolor{blue}{27} & \textcolor{blue}{27} & \textcolor{blue}{28} &  \textcolor{red}{341} \\
        \hline
                \multirow{4}{*}{\makecell{$\alpha=0.4$ }}  & Ave. Travel time & 15.4 &  15.46 & 15.53  & 15.01 \\
        \cline{2-6}
        & Ave. $\frac{1}{2} u^2$& 18.04 & 18.13 & 18.22 &  17.67\\
        \cline{2-6}
        & Ave. Fuel consumption & 53.155 & 52.77 & 52.42 &  56.5\\
        \cline{2-6}
        &Computation load (Num of QPs solved)    & 54\% (14089) & 53\% (14072) & 49\% (13573) & 100\% (27412)\\
        \cline{2-6}
        & Num of infeasible cases   & \textcolor{blue}{25} & \textcolor{blue}{25} & \textcolor{blue}{25} &   \textcolor{red}{321} \\
        \hline       
                \multirow{4}{*}{\makecell{$\alpha=0.5$ }}  & Ave. Travel time & 15.05  & 15.11 & 15.17 & 14.63 \\
        \cline{2-6}
        & Ave. $\frac{1}{2} u^2$ &24.94 & 24.88 & 24.93 & 25.08\\
        \cline{2-6}
        & Ave. Fuel consumption & 53.65 & 53.41 & 53.21 &  56.93 \\
        \cline{2-6}
        &Computation load (Num of QPs solved) & 51\% (13764)  & 51\% (13758) & 50\% (13415) &  100\% (26726) \\
        \cline{2-6}
        & Num of infeasible cases   & \textcolor{blue}{20} & \textcolor{blue}{20} & \textcolor{blue}{20} &   \textcolor{red}{341}\\
        \hline        
        \end{tabular}
        \caption{CAV metrics under event-triggered (see section III.A) and time-driven control. }
        \label{Table event}
\end{table*}

\begin{table*}\scriptsize
        \centering
        \begin{tabular}{|c|c|c|c|c|c|c|c|}
            \cline{1-8}
             &Item & \multicolumn{4}{|c|}{Self-Triggered} & Time-driven & Time-driven\\
             & & \multicolumn{4}{|c|}{}  & Modified CBF & \\
            \cline{2-8}
            & $T_{\max}$ & $0.5$ & $1$ & $1.5$ & $2$ &  $T_s=T_d = 0.05$ & $T_s = 0.05$\\
        \hline  
        \multirow{5}{*}{{$\alpha=0.1$ }} & Ave. Travel time & 19.5 & 19.48 & 19.48& 19.49&  19.5 &19.42\\
        \cline{2-8}
        & Ave. $\frac{1}{2} u^2$ & 4.27 & 5.00 & 5.93 & 7.2 & 3.37 & 3.18\\
        \cline{2-8}
        & Ave. Fuel consumption & 31.86 & 32.21 & 32.64 & 33.23 &  31.32 &31.61\\
        \cline{2-8}
                & Computation load (Num of QPs solved)& 20.46\% (7252)&  11.9\% (4218) & 10.87\% (3854) &10.32\%(3658)&  100.5\% (35636) &100\%  \\
        \cline{2-8}
        & Num of infeasible cases   & \textcolor{blue}{42}  & \textcolor{blue}{42} & \textcolor{blue}{43} & \textcolor{blue}{32}& \textcolor{red}{190} & \textcolor{red}{315}\\
         \hline
        \multirow{5}{*}{{$\alpha=0.25$ }} & Ave. Travel time & 15.57 & 15.56 & 15.57 &15.62 &15.58&15.44\\
        \cline{2-8}
        & Ave. $\frac{1}{2} u^2$& 14.33 & 15.10 & 15.68 & 16.68 & 13.38 &13.34\\
        \cline{2-8}
        & Ave. Fuel consumption & 54.45 & 53.51  & 52.57 & 52.94 &  54.17 &55.81 \\
        \cline{2-8}
                &Computation load (Num of QPs solved) &19.5\%  (5495) & 13.68\% (3857) & 12.34\% (3479) & 12.72\% (3588) & 100.9\% (28461) & 100\%(28200)\\
         \cline{2-8}
         & Num of infeasible cases   & \textcolor{blue}{27} & \textcolor{blue}{27} & \textcolor{blue}{28} &\textcolor{blue}{24} & \textcolor{red}{249} &\textcolor{red}{341} \\
        \hline
                \multirow{5}{*}{{$\alpha=0.4$ }} & Ave. Travel time & 15.15 &  15.15 & 15.18  &15.2&  15.16 &15.01 \\
        \cline{2-8}
        & Ave. $\frac{1}{2} u^2$& 18.5 & 19.32 & 19.73 & 20.36 & 17.64 & 17.67\\
        \cline{2-8}
        & Ave. Fuel consumption & 55.23 & 53.35 & 52.67 & 52.95& 54.93 &56.5\\
        \cline{2-8}
                & Computation load (Num of QPs solved)    & 20.4\% (5591) & 14.85\% (4071) &13.69\%  (3754) & 13.60\% (3727) & 101.0 \%  (27695)  &100\%  (27412)\\
        \cline{2-8}
       & Num of infeasible cases   & \textcolor{blue}{25} &\textcolor{blue}{25}& \textcolor{blue}{25} & \textcolor{blue}{20} & \textcolor{red}{220} &\textcolor{red}{321} \\
        \hline   
                \multirow{5}{*}{{$\alpha=0.5$}} & Ave. Travel time & 14.79 & 14.79 & 14.82 & 14.89&  14.8 &14.63 \\
        \cline{2-8}
        & Ave. $\frac{1}{2} u^2$ &25.5 & 25.84  & 26.43 & 27.5  &24.86 &25.08\\
        \cline{2-8}
        & Ave. Fuel consumption & 55.5 & 53.15 & 52.9 & 53.45 	 & 55.5 &56.93 \\
        \cline{2-8}
                &Computation load (Num of QPs solved) &  21.8\% (5841) &16.7\%  (4322) & 15.09\%  (4034) & 15.17\%  (4054) &  101.1\% (27033) & 100\%(26726) \\
       \cline{2-8}
       & Num of infeasible cases   &\textcolor{blue}{19} &\textcolor{blue}{20}  &  \textcolor{blue}{20}& \textcolor{blue}{20} &  \textcolor{red}{250} & \textcolor{red}{341}\\
        \hline        
        \end{tabular}
        \caption{CAV metrics under self-triggered  (see section III.B) and time-driven control.  }
        \label{Table self}
\end{table*}
We can also visualize the results presented in Tables \ref{Table event} and \ref{Table self}  by showing the variation of the average objective functions in \eqref{eqn:energyobja_m} with respect to $\alpha$ for different choices of $[s_x,s_v]$ and $T_{\max}$ respectively. As seen in Fig. \ref{fig:Objective function comparison}, by selecting higher values for bounds in the event-triggered scheme and for $T_{\max}$ in the self-triggered scheme (being more conservative) the objective functions will also attain higher values, while the lowest cost (best performance) is reached under time-driven control.
\begin{figure}
\centering
\includegraphics[scale=0.38]{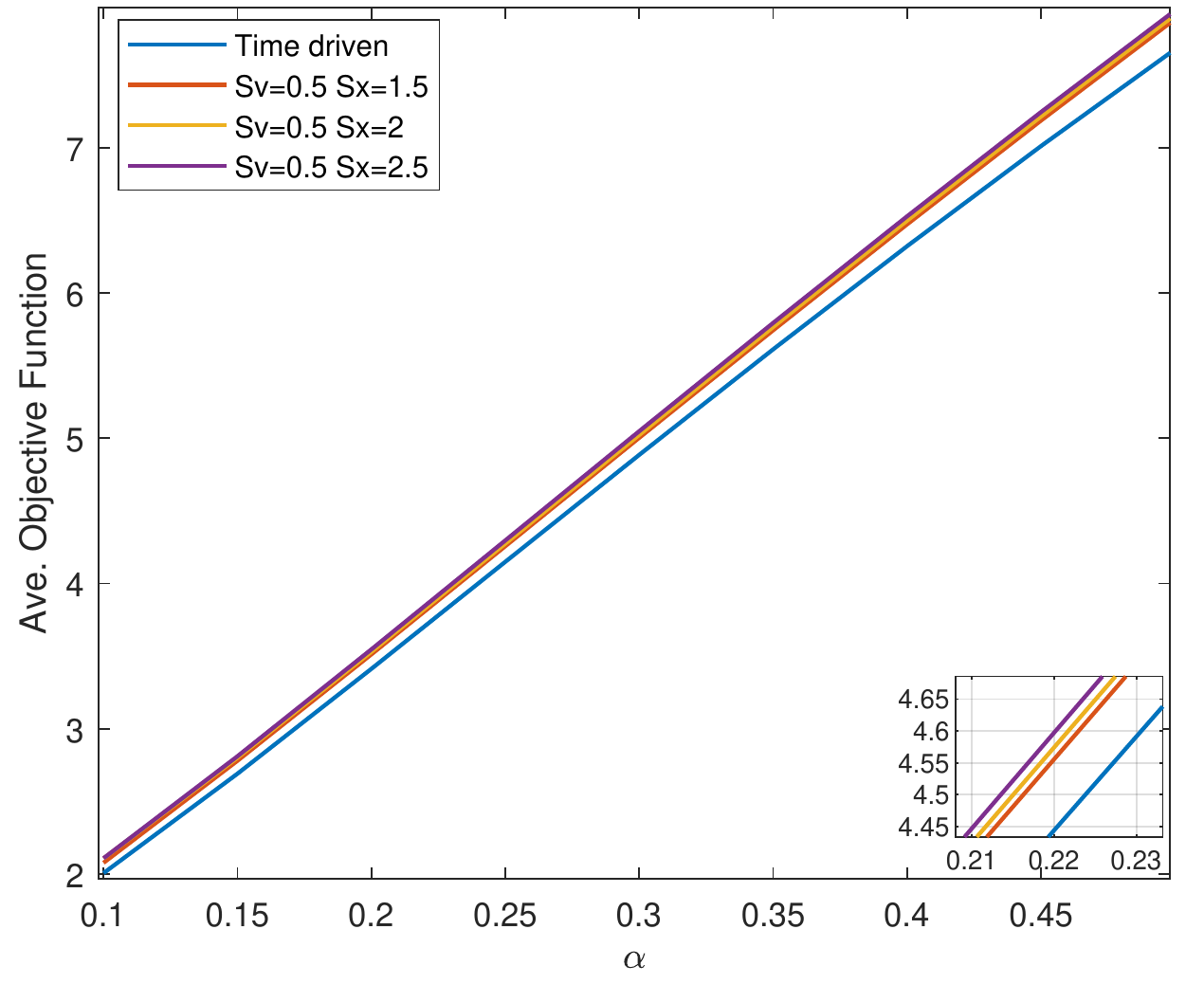} \caption{Average objective function value with respect to $\alpha$ (time weight with respect to energy in \eqref{eqn:energyobja_m}) for different selection of bounds in event-triggered approach (see Sec III.A).}
\label{fig:Objective function comparison}
\end{figure}
\begin{figure}
\centering
\includegraphics[scale=0.38]{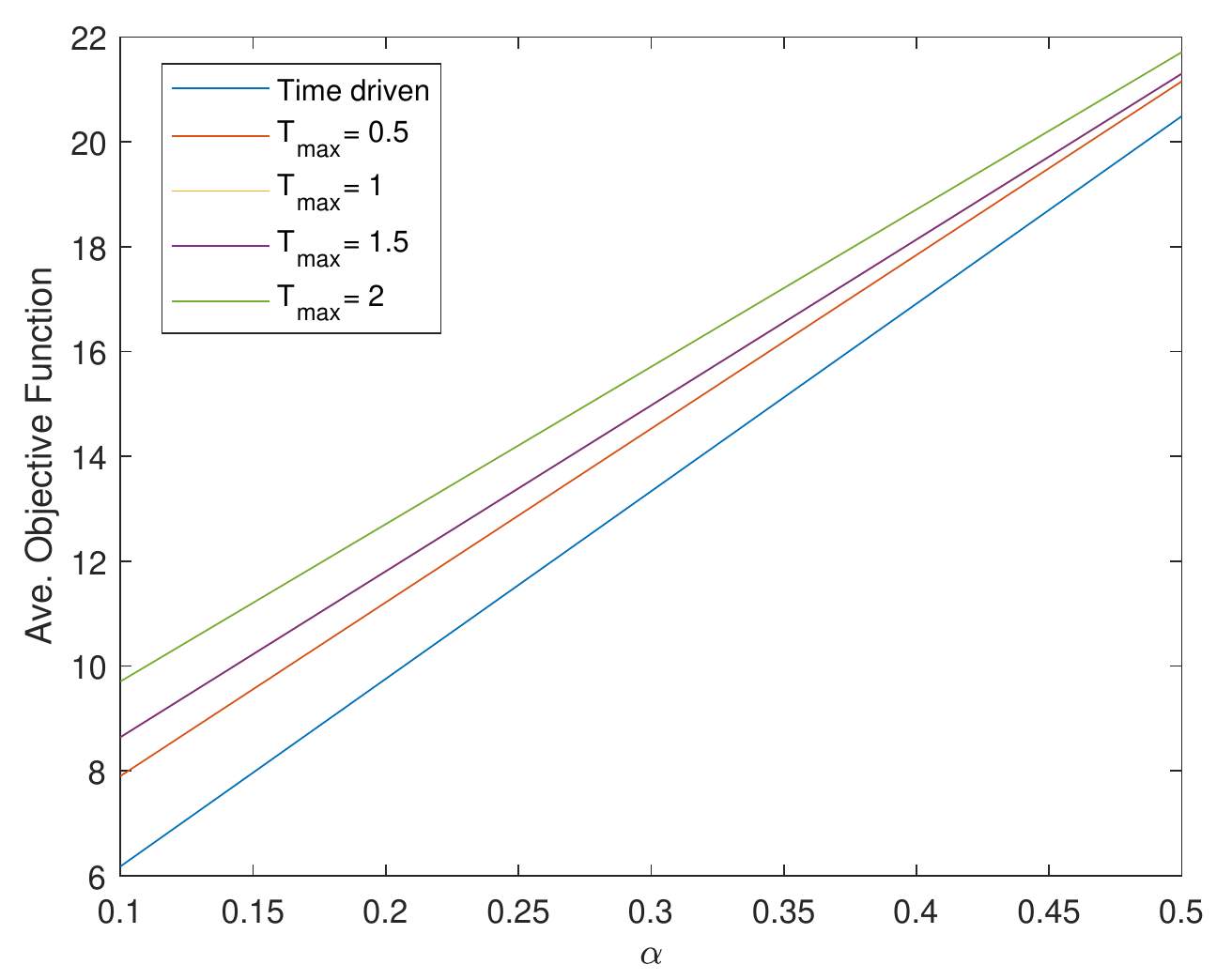} \caption{Average objective function value with respect to $\alpha$ (time weight with respect to energy in \eqref{eqn:energyobja_m}) for different selection of bounds in self-triggered approach (see Sec III.B).}
\label{fig:Objective function comparison}
\end{figure}

{\bf Constraint violation}.
It is worth noting that an ``infeasible'' QP does not necessarily imply a constraint violation, since violating a CBF constraint does not always imply the violation of an original constraint in \eqref{Safety}, \eqref{SafeMerging}, and \eqref{VehicleConstraints1}. This is due to the conservative nature of a CBF whose intent is to 
\emph{guarantee} the satisfaction of our original constraints. 
In order to explicitly show how an infeasible case may lead to a constraint violation and how this can be alleviated by the event-triggered and self-triggered schemes, we simulated 12 CAVs in the merging framework of Fig. \ref{fig:merging} with the exact same parameter settings as before and with $S=[0.5,1.5]$ in the event-triggered scheme, $T_{\max}=1$ in the self-triggered scheme and $\beta = 5$. Figure \ref{fig:rear_end} shows the values of the rear-end safety constraint over time. One can see that the satisfaction of safety constraints is always guaranteed with the event-triggered and self-triggered approach as there is no infeasible case and the value of the constraint $b_1(\textbf{x}(t)))$ is well above zero. In contrast, we see a clear violation of the constraint in the time-driven scheme in the cases of CAVs 8 depicted by the blue line.

{\bf Robustness}.
We have investigated the robustness of both schemes with respect to different forms of uncertainty, such as modeling and computational errors, by adding two noise terms to the vehicle dynamics: $\dot{x}_{i}(t) = v_{i}(t)+w_1(t)$, $\dot{v}_{i}(t) = u_{i}(t)+w_2(t)$,
where $w_1(t),w_2(t)$ denote two random processes defined in an appropriate probability space which, in our simulation, are set to be uniformly distributed over $[-2,2]$ and $[-0.2,0.2]$, respectively. We repeated the prior simulation experiment with added noise and results shown in Figs. \ref{fig:rear_end_noisy} and \ref{fig:Lateral_noisy}. We can see that the event-triggered and self-triggered schemes with almost similar performance because of their conservativeness keep the functions well away from the unsafe region (below $0$)
in contrast to the time-driven approach where we observe constraint violations due to noise, e.g., CAV 8 in Fig. \ref{fig:rear_end_noisy} and CAVs 3, 4, and, 9 in Fig. \ref{fig:Lateral_noisy}.

\begin{figure}[t]
\centering
\includegraphics[scale=0.42]{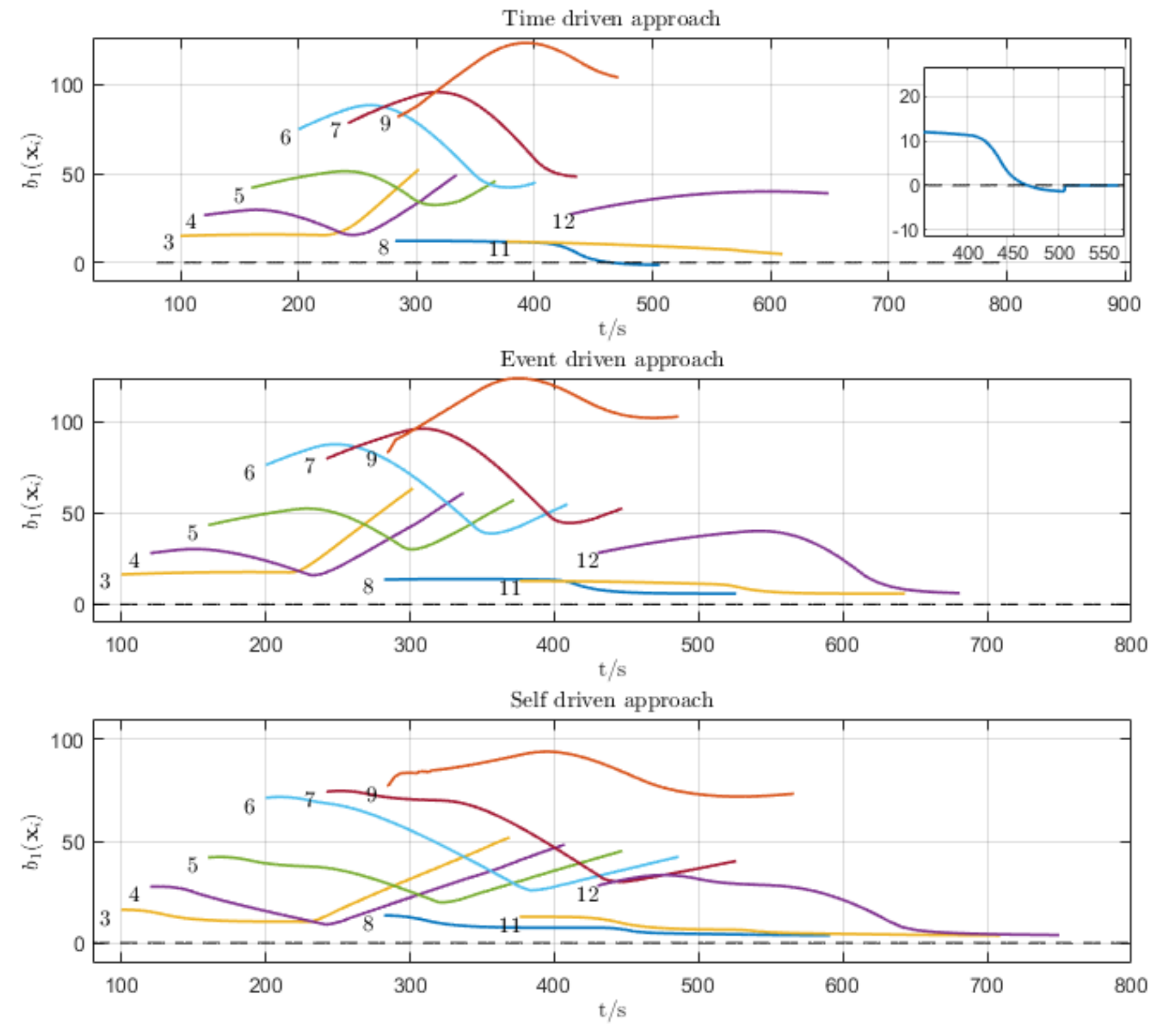} \caption{The variation of rear-end safety constraints for the time-driven, event-triggered and self-triggered approaches.}
\label{fig:rear_end}
\end{figure}

\begin{figure}[t]
\centering
\includegraphics[scale=0.45]{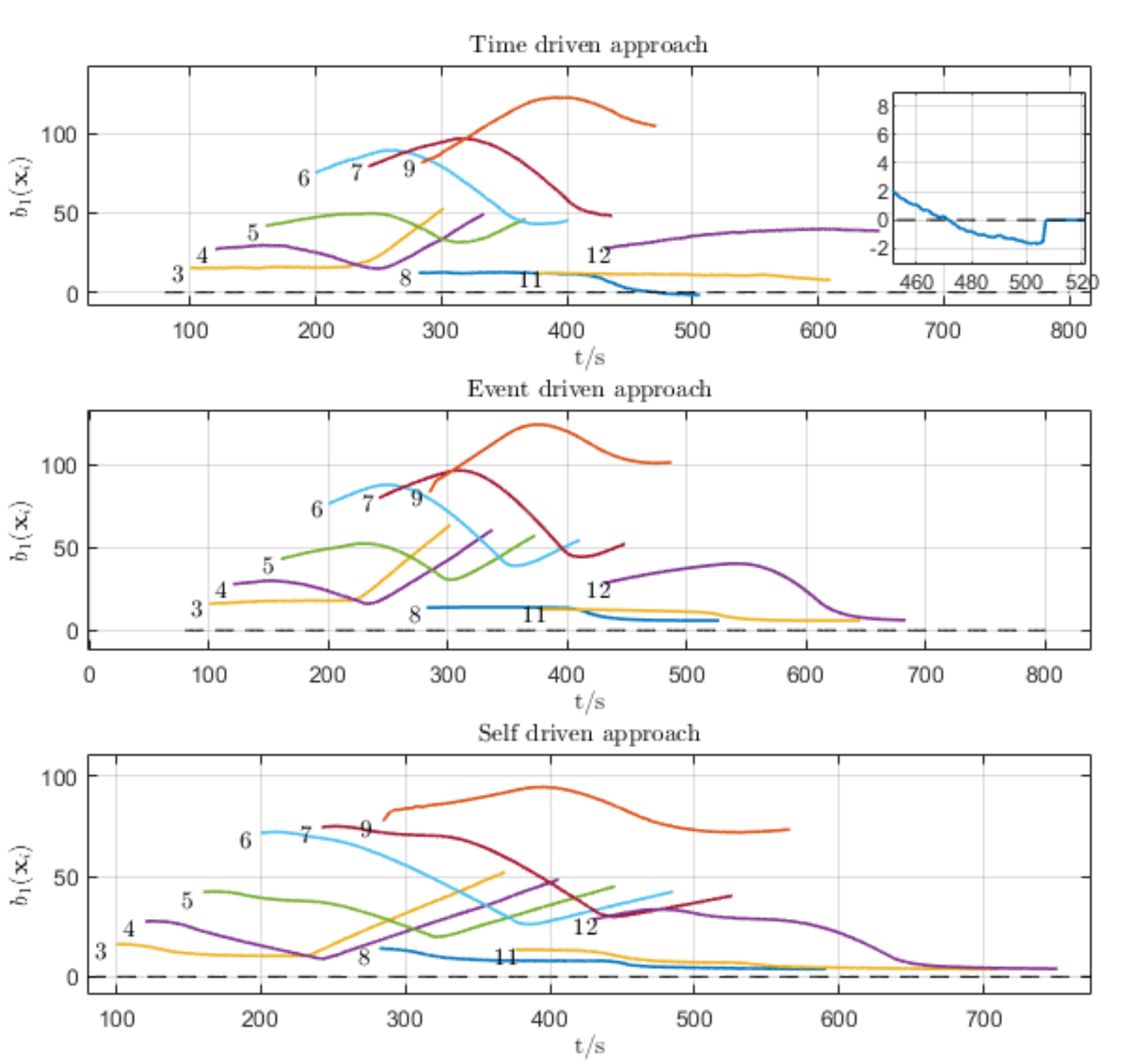} \caption{The variation of rear-end safety constraints for the time-driven, event-triggered, and self-triggered approaches in the presence of noise. Note the constraint violation under time-driven control.}%
\label{fig:rear_end_noisy}
\end{figure}

\begin{figure}[t]
\centering
\includegraphics[scale=0.45]{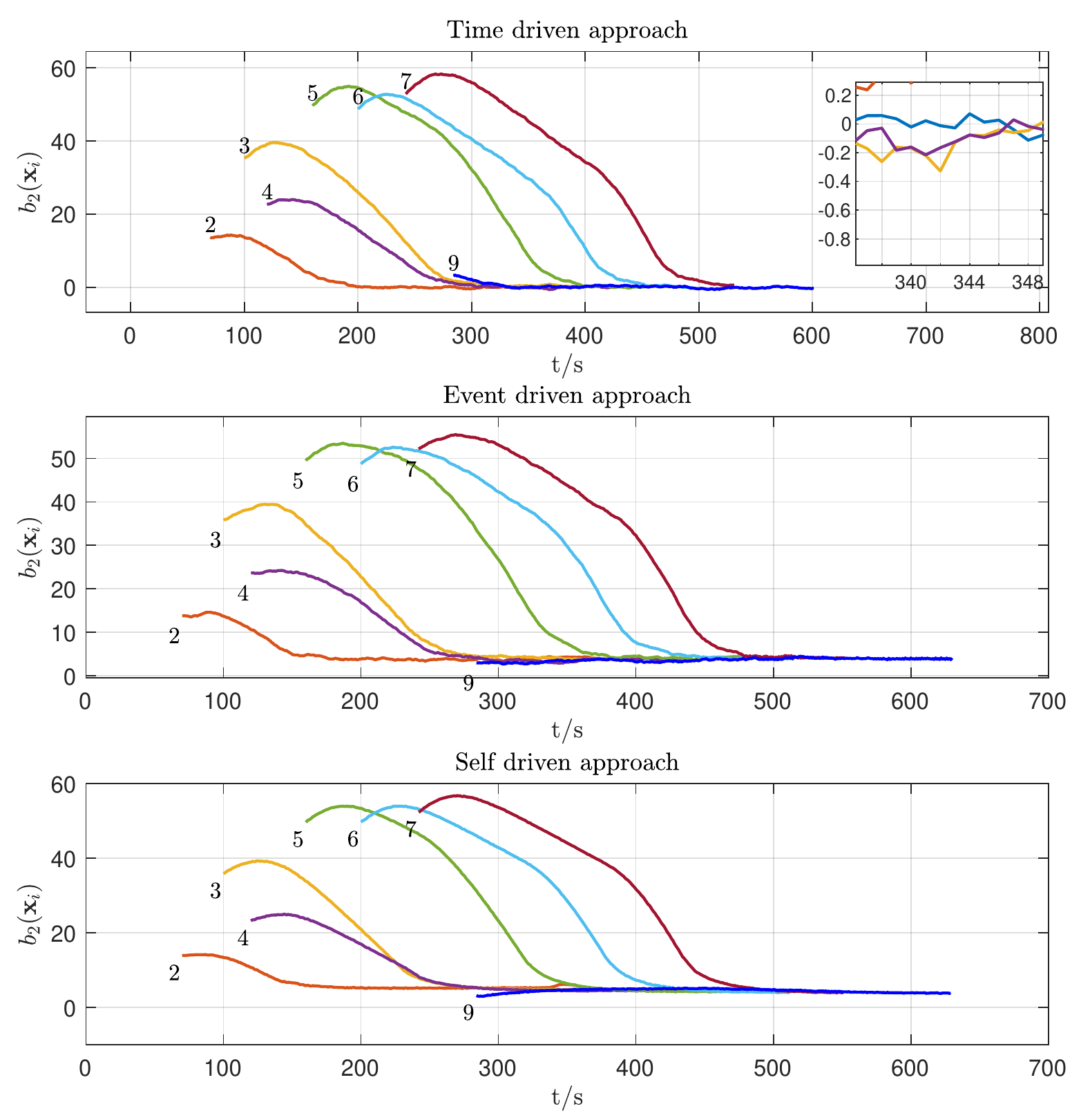} \caption{The variation of safe merging constraints for the time-driven, event-triggered and self-triggered approaches in the presence of noise. Note the constraint violations under time-driven control.}%
\label{fig:Lateral_noisy}
\end{figure}


\section{CONCLUSIONS}
The problem of controlling CAVs in conflict areas of a traffic network subject to hard safety constraints can be solved through a combination of tractable optimal control problems and the use of CBFs. These solutions can be derived by discretizing time and solving a sequence of QPs. However, the feasibility of each QP cannot be guaranteed over every time step. When this is due to the lack of a sufficiently high control update rate, we have shown that this problem can be alleviated through either an event-triggered scheme or a self-triggered scheme, while at the same time reducing the need for communication among CAVs, thus lowering computational costs and the chance of security threats. 
Ongoing work is targeted at eliminating all possible infeasibilities through the use of sufficient conditions based on the work in \cite{XIAO2022inf} added to the QPs, leading to complete solutions of CAV control problems with full safety constraint guarantees. 

\label{sec:conclude}





\begin{thebibliography}{10}

\bibitem{li2013survey}
D.~Wen L.~Li and D.~Yao.
\newblock A survey of traffic control with vehicular communications.
\newblock {\em IEEE Trans. on Intelligent Transportation Systems}, 15(1):pp.
  425--432, 2013.

\bibitem{9625017}
Cong Gao, Geng Wang, Weisong Shi, Zhongmin Wang, and Yanping Chen.
\newblock Autonomous driving security: State of the art and challenges.
\newblock {\em IEEE Internet of Things Journal}, 9(10):7572--7595, 2022.

\bibitem{deWaard09}
D.~{de Waard}, C.~Dijksterhuis, and KA. Brookhuis.
\newblock Merging into heavy motorway traffic by young and elderly drivers.
\newblock {\em Accident Analysis \& Prevention}, 41(3):pp. 588--597, 2009.

\bibitem{Schrank20152015UM}
T.~Lomax D.~Schrank, B.~Eisele and J.~Bak.
\newblock 2015 urban mobility scorecard.
\newblock 2015.

\bibitem{kavalchuk2020performance}
Ilya Kavalchuk, A~Kolbasov, K~Karpukhin, A~Terenchenko, et~al.
\newblock The performance assessment of low-cost air pollution sensor in city
  and the prospect of the autonomous vehicle for air pollution reduction.
\newblock In {\em IOP Conference Series: Materials Science and Engineering},
  volume 819, page 012018. IOP Publishing, 2020.

\bibitem{VANDENBERG201643}
Vincent~A.C. {van den Berg} and Erik~T. Verhoef.
\newblock Autonomous cars and dynamic bottleneck congestion: The effects on
  capacity, value of time and preference heterogeneity.
\newblock {\em Transportation Research Part B: Methodological}, 94:43--60,
  2016.

\bibitem{7562449}
J.~Rios-Torres and A.~A. Malikopoulos.
\newblock A survey on the coordination of connected and automated vehicles at
  intersections and merging at highway on-ramps.
\newblock {\em IEEE Trans. on Intelligent Transportation Systems}, 18(5):pp.
  1066--1077, 2017.

\bibitem{xu2019grouping}
H.~Xu, S.~Feng, Y.~Zhang, and L.~Li.
\newblock A grouping-based cooperative driving strategy for cavs merging
  problems.
\newblock {\em IEEE Trans. on Vehicular Technology}, 68(6):pp. 6125--6136,
  2019.

\bibitem{wu2013mathematical}
Jia Wu, Fei Yan, and Abdeljalil Abbas-Turki.
\newblock Mathematical proof of effectiveness of platoon-based traffic control
  at intersections.
\newblock In {\em 16th International IEEE Conference on Intelligent
  Transportation Systems (ITSC 2013)}, pages 720--725. IEEE, 2013.

\bibitem{rajamani2000demonstration}
Rajesh Rajamani, Han-Shue Tan, Boon~Kait Law, and Wei-Bin Zhang.
\newblock Demonstration of integrated longitudinal and lateral control for the
  operation of automated vehicles in platoons.
\newblock {\em IEEE Transactions on Control Systems Technology}, 8(4):695--708,
  2000.

\bibitem{1373519}
K.~Dresner and P.~Stone.
\newblock Multiagent traffic management: a reservation-based intersection
  control mechanism.
\newblock In {\em Proceedings of the Third International Joint Conference on
  Autonomous Agents and Multiagent Systems, 2004. AAMAS 2004.}, pages 530--537,
  2004.

\bibitem{au2010motion}
Tsz-Chiu Au and Peter Stone.
\newblock Motion planning algorithms for autonomous intersection management.
\newblock In {\em Workshops at the Twenty-Fourth AAAI Conference on Artificial
  Intelligence}, 2010.

\bibitem{zhang2013analysis}
Kailong Zhang, Arnaud De~La~Fortelle, Dafang Zhang, and Xiao Wu.
\newblock Analysis and modeled design of one state-driven autonomous
  passing-through algorithm for driverless vehicles at intersections.
\newblock In {\em 2013 IEEE 16th International Conference on Computational
  Science and Engineering}, pages 751--757. IEEE, 2013.

\bibitem{7313484}
J.~Rios-Torres, A.~Malikopoulos, and P.~Pisu.
\newblock Online optimal control of connected vehicles for efficient traffic
  flow at merging roads.
\newblock In {\em 2015 IEEE 18th International Conf. on Intelligent
  Transportation Systems}, pages 2432--2437. IEEE, 2015.

\bibitem{XIAO2021109333}
W.~Xiao and C.G. Cassandras.
\newblock Decentralized optimal merging control for connected and automated
  vehicles with safety constraint guarantees.
\newblock {\em Automatica}, 123:109333, 2021.

\bibitem{Zhang2018}
Y.F Zhang and C.G. Cassandras.
\newblock Decentralized optimal control of connected automated vehicles at
  signal-free intersections including comfort-constrained turns and safety
  guarantees.
\newblock {\em Automatica}, 109:p. 108563, 11 2019.

\bibitem{garcia1989model}
Carlos~E Garcia, David~M Prett, and Manfred Morari.
\newblock Model predictive control: Theory and practice—a survey.
\newblock {\em Automatica}, 25(3):335--348, 1989.

\bibitem{cao2015cooperative}
W.~Cao, M.~Mukai, T.~Kawabe, H.~Nishira, and N.~Fujiki.
\newblock Cooperative vehicle path generation during merging using model
  predictive control with real-time optimization.
\newblock {\em Control Engineering Practice}, 34:98--105, 2015.

\bibitem{Xiao2019}
W.~Xiao and C.~Belta.
\newblock Control barrier functions for systems with high relative degree.
\newblock In {\em Proc. of 58th IEEE Conf. on Decision and Control}, pages
  474--479, Nice, France, 2019.

\bibitem{CBF_QP(2017)}
A.~Ames, X.~Xu, J.W. Grizzle, and P.~Tabuada.
\newblock Control barrier function based quadratic programs for safety critical
  systems.
\newblock {\em IEEE Trans. on Automatic Control}, 62(8):3861--3876, 2017.

\bibitem{mukai2017model}
M.~Mukai, H.~Natori, and M.~Fujita.
\newblock Model predictive control with a mixed integer programming for merging
  path generation on motor way.
\newblock In {\em 2017 IEEE Conf. on Control Technology and Applications
  (CCTA)}, pages 2214--2219. IEEE, 2017.

\bibitem{XIAO2021109592}
W.~Xiao, C.G. Cassandras, and C.~Belta.
\newblock Bridging the gap between optimal trajectory planning and
  safety-critical control with applications to autonomous vehicles.
\newblock {\em Automatica}, 129:109592, 2021.

\bibitem{ong2018event}
P.~Ong and J.~Cort{\'e}s.
\newblock Event-triggered control design with performance barrier.
\newblock In {\em 2018 IEEE Conf. on Decision and Control (CDC)}, pages
  951--956. IEEE, 2018.

\bibitem{taylor2020safety}
Taylor~A. J., Ong P., Cortés J., and Ames~A. D.
\newblock Safety-critical event triggered control via input-to-state safe
  barrier functions.
\newblock {\em IEEE Control Systems Letters}, 5(3):749--754, 2020.

\bibitem{Xiao2021EventTriggeredSC}
W.~Xiao, C.~Belta, and C.~G. Cassandras.
\newblock Event-triggered safety-critical control for systems with unknown
  dynamics.
\newblock In {\em 2021 60th IEEE Conf. on Decision and Control (CDC)}, pages
  540--545, 2021.

\bibitem{XIAO2022inf}
W.~Xiao, C.~A. Belta, and C.~G. Cassandras.
\newblock Sufficient conditions for feasibility of optimal control problems
  using control barrier functions.
\newblock {\em Automatica}, 135:109960, 2022.

\bibitem{Vogel2003}
K.~Vogel.
\newblock A comparison of headway and time to collision as safety indicators.
\newblock {\em Accident Analysis \& Prevention}, 35(3):427--433, 2003.

\bibitem{2020Weidynreseq}
Wei Xiao and Christos~G. Cassandras.
\newblock Decentralized optimal merging control for connected and automated
  vehicles with optimal dynamic resequencing.
\newblock In {\em 2020 American Control Conference (ACC)}, pages 4090--4095,
  2020.

\bibitem{kamal2012model}
M.~A.~S. Kamal, M.~Mukai, J.~Murata, and T.~Kawabe.
\newblock Model predictive control of vehicles on urban roads for improved fuel
  economy.
\newblock {\em IEEE Trans. on Control Systems Technology}, 21(3):831--841,
  2012.

\end{thebibliography}

\end{document}